\documentclass[conference]{IEEEtran}
\IEEEoverridecommandlockouts
\usepackage{cite}
\usepackage{amsmath,amssymb,amsfonts}
\usepackage{algorithmic}
\usepackage{graphicx}
\usepackage{subfigure}
\usepackage{caption}
\usepackage{textcomp}
\usepackage{epstopdf}
\usepackage{xcolor}
\usepackage{amsthm}
\usepackage{amsmath}
\usepackage{amssymb}
\usepackage{array}
\newtheorem{definition}{Definition} [section]
\newtheorem{theorem}{Theorem}  [section]
\def\BibTeX{{\rm B\kern-.05em{\sc i\kern-.025em b}\kern-.08em
    T\kern-.1667em\lower.7ex\hbox{E}\kern-.125emX}}
\begin{document}

\title{An Uncertainty-  and Collusion-Proof  Voting Consensus Mechanism in Blockchain\\
}
\author{
\IEEEauthorblockN{Shengling Wang \IEEEauthorrefmark{1}, Xidi Qu \IEEEauthorrefmark{1}, Qin Hu \IEEEauthorrefmark{2}, Weifeng Lv(Corresponding Author)\IEEEauthorrefmark{3}}
\IEEEauthorblockA{
\IEEEauthorrefmark{1}\textit{School of Artificial Intelligence, Beijing Normal University, Beijing, China}\\
Email:{wangshengling@bnu.edu.cn, sindychanson@mail.bnu.edu.cn} \\
\IEEEauthorrefmark{2}\textit{Department of Computer and Information Science, Indiana University-Purdue University Indianapolis, Indiana, USA} \\
Email:qinhu@iu.edu \\
\IEEEauthorrefmark{3}\textit{School of Computer Science and Engineering, Beihang University, Beijing, China} \\
Email:lwf@nlsde.buaa.edu.cn
}}

\maketitle

\begin{abstract}
Though voting-based consensus algorithms in Blockchain outperform  proof-based ones   in energy- and transaction-efficiency, they are prone to incur wrong elections and bribery elections. The former originates from the uncertainties of candidates' capability and availability; and the latter comes from the egoism of voters and candidates.   Hence, in this paper, we propose an uncertainty- and collusion-proof voting consensus mechanism, including the selection pressure-based voting consensus algorithm and the  trustworthiness evaluation algorithm.  The first algorithm can decrease the side effects of candidates' uncertainties, lowering wrong elections while trading off the balance between efficiency and fairness in electing miners. The second algorithm adopts an  incentive compatible scoring rule to evaluate the  trustworthiness of voting, motivating voters to report true beliefs on candidates by making egoism in  consistent with altruism so as to avoid bribery elections. A salient feature of our work is theoretically analyzing the proposed voting consensus mechanism by the large deviation theory. Our analysis provides not only the voting failure rate of a candidate
but also its decay speed, based on which the concepts of {\it the effective selection valve} and {\it  the effective expectation of merit} are introduced to help the system designer to determine the optimal voting standard and guide a candidate to behave in an optimal way for lowering the voting failure rate.
\end{abstract}

\section{Introduction}
\label{into}
Featured with decentralization, transparency,
immutability and global inclusiveness,  Blockchain is breeding extensive novel applications and fields, pulling us from Internet
of information to {\it Internet of value}.  So far, the market size of Blockchain is around USD 1,907.8 million, which is estimated to reach USD 7,684 million by 2022, at a compound annual growth rate of 79.6\% \cite{Markets17}. Essentially,  Blockchain is an event\footnote{ When Blockchain is applied in the financial sector, the transaction is a typical event.}-driven  deterministic state machine, which calls for consensus algorithms to agree on the order of deterministic events and screen out invalid events. The mainstream consensus algorithms in Blockchain are proof of work (PoW)\cite{Nakamoto08} and proof of stake (PoS)\cite{King12}.

PoW allocates the accounting rights and rewards through the competition among nodes to solve cryptographic puzzles in the form of
a hash computation by brute-forcing. PoW obviously incurs a serious low-efficiency issue due to the immense scale of electricity wasting and long transaction confirmation delay.  Yet, such a way of sacrificing efficiency does not bring fairness even though it is the original design intention of PoW.
This is because mining one block is so hard that extensive individual miners tend to gather their hashing powers to form {\it mining pools} for seeking the solution of PoW puzzles together. Thus, the accounting rights and rewards are gradually concentrated in several super mining pools, leading a solo miner hardly has the opportunity to participate in the decision-making of Blockchain. This unfairness makes complete decentralization cannot be realized through PoW.

As a  power-saving  alternative to PoW, PoS
 introduces a concept of the coin age,
 which is the
value of the coin multiplied by the time period after the coin was created. PoS sets the difficulty of mining to be inversely proportional to the
coin age. Thus,  the consensus in Blockchain is no longer
completely relying on the proof of work, effectively addressing the problem of
electricity wasting. However, PoS still has a limitation in improving efficiency since it does not significantly lower the transaction confirmation delay. Not only that, PoS leads to more unfairness in that the rich miners are bound to be dominant in the Blockchain network.

Different from PoW and PoS where all nodes participate in the decision making of Blockchain,  delegated PoS (DPoS) \cite{Larimer14} leverages the voting power of stakeholders to select qualified nodes to mine, who will receive rewards when creating correct and timely blocks.
 Since no competition among miners in mathematical computations, DPoS outperforms PoW and  PoS in energy- and transaction-efficiency. However, DPoS suffers two issues: 1) {\it wrong election}. Due to the information asymmetry,  a voter is hardly fully aware of a candidate's behavior. When the candidate mines precariously or behaves maliciously, he\footnote{In this paper, we use ``he" and ``she" to differentiate the candidate and the voter.} will be voted out at the end of one round. However, such an ex-post countermeasure cannot compensate for the loss of the voters caused by their wrong selection in this round; 2) {\it bribery election}.  Though the voting mechanism seems democratic,  it can easily ruin fairness in practice since voters are also stakeholders, which makes room for the voters and candidates to collude, leading to bribery election.

In this paper, we propose an uncertainty-  and collusion-proof voting consensus mechanism, which inherits the advantages of DPoS in terms of energy- and transaction-efficiency, but has no trouble in  wrong and bribery elections.  It is challenging to realize the proposed voting consensus mechanism in that 1) the uncertainty leading to wrong elections is dual: one from unknown candidates and the other from familiar ones. In detail,  lacking knowledge on the mining ability of unknown candidates, a voter may miss professional miners or mistakenly vote incompetent candidates as qualified ones.  Even, the voter may still make wrong decisions when facing candidates who  had been selected before due to their uncertain behaviors,  such as
becoming unavailable because of electricity shortage or  malicious
attacks; 2)
the real intention of voting (egoism or altruism) is the private information of a voter, which makes it easy to hide collusion between voters and candidates, implying that bribery elections are prone to take place.

To tackle the first challenge,  we propose a selection pressure-based voting consensus algorithm. In ecology, the selection pressure is used to characterize the advantage that a creature's traits are chosen to survive. We introduce this concept in the proposed voting consensus algorithm to depict the urgency of a candidate being chosen. The selection pressure is designed to elect miners with high availability and less frequency of being chosen before. Such a design has two merits: 1) it can not only reduce side effects from uncertain behaviors of known candidates but also explore unfamiliar or even completely unknown candidates, decreasing wrong elections due to the low cognitive level of global information; 2) it can trade off the balance between efficiency and fairness in electing miners.

To address the second challenge, we adopt a  scoring rule for evaluating the trustworthiness of voting. Such a rule subjects to the incentive compatibility (IC) constraint, so that a rational and intelligent voter finds she  can obtain a higher score when telling the truth rather than lying. Thus, the motivation of collusion between the voter and the candidate is naturally eliminated, implying that bribery elections can be prevented.    An attractive nature of the trustworthiness evaluation algorithm is that there is no need to submit any private information. Instead, a voter only needs to report  her opinion of  any other random peer voter's prior and posterior
 beliefs on candidates.
  The reason behind such a design is each voter's subjective belief comes from the same Blockchain system, which affects every voter in the same way. Hence, the subjective belief of any voter is symmetric with regard to that of any other voter.

A salient feature of our work is theoretically analyzing  the proposed  voting consensus mechanism by taking advantage of the large deviation theory. The main contributions and results of this part are summarized as follows:
\begin{itemize}
  \item  We provide not only the voting failure rate of a candidate  but also its decay speed. Such a  fine-grained  analysis reveals that to slow down the growth rate of the voting failure rate,
lowering the selection criteria is more effective than increasing the selection pressure of a candidate.
  \item  A concept of the effective selection valve is introduced to describe the lowest selection criteria that can  guarantee the voting failure rate of a candidate  no more than a given tolerance
degree. Through solving the effective selection valve,  the optimal  number of super nodes (i.e., the winners in the voting) can be determined for the system
designer to meet the requirement of the voting failure
rate, essentially ensuring the fairness of the
Blockchain system.
        \item  A concept of the effective
expectation of merit  is introduced as the minimum expected merit a candidate brought to the voter, which makes his voting failure rate no more than
a given tolerance degree. The value of this metric quantitatively answers a question:
how much effort should a candidate devote to become a super node?
\end{itemize}

The remaining part of the paper proceeds as follows:
In Section \ref{sec:formulation}, we present an overview of our proposed voting consensus mechanism. The main part of the mechanism, named selection pressure-based voting consensus algorithm, is elaborated in Section \ref{sec:pressure}, which is further enhanced with the voting trustworthiness evaluation in Section \ref{sec:trustworthiness}. We conduct an experimental evaluation in Section \ref{sec:experiment} to illustrate the effectiveness of our mechanism, and summarize the most related work in Section \ref{sec:related}. The whole paper is concluded in Section \ref{sec:conclusion}.

\section{Overview of Our  Voting Consensus Mechanism}
\label{sec:formulation}

In this section, an overview of our proposed uncertainty-  and collusion-proof voting consensus mechanism is introduced, where $N$ voters are required select $K$ super nodes  from $M$ candidates as representatives to mine.  Since a node obtains the accounting right through the voting rather than   the competition with other nodes in terms of computing power, the proposed mechanism inherits the advantage of  DPoS in saving energy  and speeding up  transactions. However, such a voting-based consensus mechanism  may suffer the issues of wrong and bribery elections as mentioned in Section \ref{into}, pressing a need to find solutions.

To avoid bribery  and wrong elections, we employ {\it score}  to evaluate the  capability  of any candidate to be  a super node. $K$ nodes with the highest scores are selected as super nodes. The score  $S_j^k$ of any   candidate $j$ in the $k^{th}$ round of voting  can be  calculated  as
\begin{equation}\label{score}
S_j^k= \sum_{i=1}^ N s_{i} \times t_{ij}^k \times c_{ij}^k,\ \ \ j=1,2,\cdots,M.
\end{equation}
In \eqref{score},  $c_{ij}^k \in\{0,1\}$ is an indicator identifying whether voter $i$  chooses candidate $j$ as a super node in round $k$, satisfying $\sum_{j=1}^{M}{c_{ij}^k\leq K}$, which is a key parameter needed to be controlled strategically to refrain from wrong selection. 
$t_{ij}^k$ represents the  trustworthiness of votes that voter $i$ casts to $j$ in round $k$, reflecting the degree that voter $i$ chooses candidate $j$ in light of his block mining capability rather than other factors, which is designed to avoid bribery election. The larger the value of  $t_{ij}^k$ is, the more reliable $i$ chooses $j$, leading to his high score. Finally, $s_{i}$ is the stake of voter $i$, indicating her voting weight or speaking power  in voting process.

In the proposed uncertainty-  and collusion-proof voting consensus mechanism,  as shown in Fig. \ref{chart},  to reduce wrong selections,  we propose the selection pressure-based voting  consensus algorithm to determine the optimal $c_{ij}^k$ ($i=1,2,\cdots,N; \  j=1,2,\cdots,M$), combined with which, the  voting trustworthiness $t_{ij}^k$ can be calculated through the evaluation algorithm to  hold back bribery elections. We will detail the above two algorithms in the following sections.

\begin{figure}[ht]
\centerline{
\includegraphics[width=9cm]{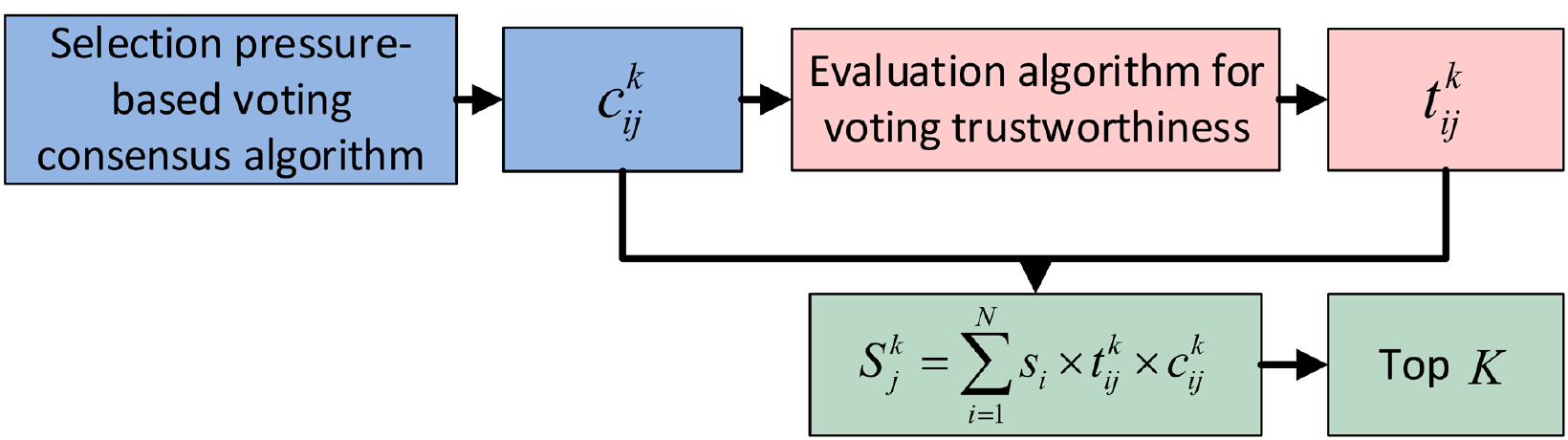}}
\caption{Overview of the proposed voting consensus mechanism.}
\label{chart}
\end{figure}

\section{Selection Pressure-Based Voting  Consensus Algorithm and Analysis}
\label{sec:pressure}
In this section, we will introduce the proposed selection pressure-based voting  consensus algorithm and then take advantage of the large deviation theory \cite{ganesh2004big} to analyze it.
\subsection{Selection Pressure-Based Voting Consensus Algorithm}
In the proposed  selection pressure-based voting consensus algorithm, the selection pressure  $F_{ij}^t$ of each candidate $j$  ($j=1,2,\cdots,M$) in the $t^{th}$ round of voting is formulated as follow
\begin{equation}\label{wij}
F_{ij}^t=M_{ij}^{t-1}-C_{ij}^{t-1}, \ t=2,3,\cdots
\end{equation}
In \eqref{wij}, $M_{ij}^{t-1}=m_{ij}^1+m_{ij}^2+\cdots+m_{ij}^{t-1}$ and $C_{ij}^{t-1}=c_{ij}^1+c_{ij}^2+\cdots+c_{ij}^{t-1}$, where $m_{ij}^{k} \geq 0$ is the  merit of candidate $j$ to voter $i$ in round $k$ ($ k=1,2,\cdots,t-1$) and  $c_{ij}^k\in \{0,1\}$  is whether candidate $j$ is selected by voter $i$ in round $k$ as mentioned above.

  The proposed voting consensus algorithm elects the top $K$ nodes with the biggest $F_{ij}^t$ for voter $i$ to  determine $c_{ij}^t$, which can trade off the balance between  efficiency and  fairness in electing super nodes.  In detail,  the bigger  $F_{ij}^t$  implies the higher $M_{ij}^{t-1}$ and the lower $C_{ij}^{t-1}$. The greater the  merit  $M_{ij}^{t-1}$ of candidate $j$, the more the contribution  he had made to voter $i$ until round $t-1$. Selecting candidates with the higher $M_{ij}^{t-1}$ can
improve the profit of voter $i$ in the future. In addition, the lower $C_{ij}^{t-1}$ means  the fewer historical selections are made on $j$. Electing candidates with the lower $C_{ij}^{t-1}$ can reflect the fairness of the voter in the voting process, which also
 facilitates  exploring unfamiliar or completely unknown candidates. Thus, through trying to have a more comprehensive understanding of the overall information about candidates, voters can make rational and intelligent elections.

In the following, we will detail how to calculate $m_{ij}^k$ ($i=1,2,\cdots,N; \  j=1,2,\cdots,M; \ k=1,2,\cdots$) in \eqref{wij}. Since  $m_{ij}^k$ means the merit of candidate $j$ to voter $i$ in round $k$, it is affected by two factors: the availability of candidate $j$ and the profit he brought to voter $i$ in round $k$, indicated by $R_{ij}^k$.  Obviously,
a candidate with high historical contribution and low unavailability  is more valuable to a voter. Hence, $m_{ij}^k$ can be calculated as
\begin{equation}\label{aijt}
m_{ij}^k=\rho{d(u_j)}+{R_{ij}^k}.
\end{equation}
In \eqref{aijt}, $\rho$ is a scaling  parameter. $d(u_j)$ is a function indicating the availability of candidate $j$, where $u_j \in [0,1]$  is the unavailable probability of candidate $j$.  The larger $u_j$, the smaller $d(u_j)$. In practice, $u_j$ can be estimated  by the total number of unavailable times divided by the total number of times that candidate $j$ being elected. Obviously, the method using the availability of a candidate as an important metric for evaluating his merit  can effectively reduce the negative impact of the candidate's uncertain behaviors on  voting, lowering the probability of wrong election.

\subsection{Analysis Based on the Large Deviation Theory}
In this subsection, we take advantage of the large deviation theory to analyze the performance of the proposed voting consensus algorithm. To that aim, we create a virtual queue $Q_{ij}^t$ ($i=1,2,\cdots,N; \  j=1,2,\cdots,M$) as follows,
\begin{equation}\label{nwij}
Q_{ij}^t=C_{ij}^{t-1}-M_{ij}^{t-1}, \ t=2,3,\cdots.
\end{equation}
It is easy to find that $Q_{ij}^t=-F_{ij}^t$, implying that
the larger $F_{ij}^t$, the smaller $Q_{ij}^t$.
Since the proposed voting consensus algorithm elects the top $K$ nodes with the biggest $F_{ij}^t$,
$Q_{ij}^t$ actually characterizes the ranking of candidate $j$ for being elected. A candidate with a  large $F_{ij}^t$ has a high urgency to be elected, thus he ranks high, implying he has short virtual queue length. Hence, we name the virtual queue $Q_{ij}^t$ the {\it ranking queue}. A node with short length of ranking queue has high chance  to be elected.

Since any candidate $j$ being elected in each round subjects to the Binomial distribution, $C_{ij}^{t-1}$ can be assumed as a Possion process. Similarly, the merit of candidate $j$ to voter $i$ in each round depends on whether he is elected, which is an event obeying  the Binomial distribution as just mentioned, implying that $M_{ij}^{t-1}$ can also be viewed as a  Possion process.  We assume that $\lambda$ and $\Lambda$ are respectively the expectations of $M_{ij}^{t-1}$ and $C_{ij}^{t-1}$. To guarantee the stability of the  virtual queue $Q_{ij}^t$, $\lambda > \Lambda$ should hold, and thus  $\rho$ in \eqref{aijt} should satisfy $\rho> \frac{\Lambda(1-R)}{d(u_j)}$, where $R$ is the maximum of $R_{ij}^t$.

For any candidate,  his main concern is under what conditions will he fail in the election.
This is also the key in analyzing  a voting consensus algorithm. To that aim,
we study the metric $L_{ij}=\sup_{t\geq 2}Q_{ij}^t$, which  is the longest ranking queue length of candidate $j$  when $t\rightarrow \infty $, reflecting his historic lowest selection pressure.  Considering each voter can select  up to $K$ super nodes in our proposed voting consensus algorithm, we have the following definition:
\begin{definition}[Selection valve]
\label{def:valve}
The selection valve $L$ is the ranking queue length of a candidate whose selection pressure is ranked $K^{th}$.
\end{definition}

Obviously, when $L_{ij}>L$, candidate $j$ may fail in the election. Otherwise, he will be elected by voter $i$. To analyze such an event, we define the following concept:
\begin{definition}[Voting failure rate]
\label{def:loss}
The voting failure rate indicates the probability that any candidate $j$ is not elected by voter $i$ when his  historic  longest  ranking queue $L_{ij}$   is longer than the selection valve $L$, i.e., $P(L_{ij}>L)$.
\end{definition}

 In the following, we will  analyze the distribution of $P(L_{ij}>L)$.
Let $L=lb$ for $b>0$. In light of the Cramer's theorem \cite{ganesh2004big},  when $l\rightarrow \infty$, we have
\begin{equation}\label{loss}
\lim _{l\rightarrow \infty}\frac{1}{l}\log {P(L_{ij}>lb)}= -I(b),
\end{equation}
with
\begin{equation}\label{IB}
I(b)=\inf_{t \geq 2}t\psi ^*(\frac{b}{t}),
\end{equation}
where
\begin{equation}\label{psi}
\psi ^*(x)=\sup_{\theta \in{ \mathbb{R}}}\left \{\theta {x} - \psi(\theta)\right \}.
\end{equation}
In \eqref{psi}, $\psi ^*(x)$ is the convex conjugate or the Legendre transformation of $\psi(\theta )$ which is the cumulant generating function of $Q_{ij}^t$ and can be expressed as
\begin{equation}\label{psii}
\begin{split}
\psi(\theta) &=\lim _{t\rightarrow \infty}\frac{1}{t}\log { \textbf{E}} \left[ e^{\theta (Q_{ij}^t)} \right] \\
&=\lim _{t\rightarrow \infty}\frac{1}{t}\log { \textbf{E}} \left[ e^{\theta ({C_{ij}^{t-1}-M_{ij}^{t-1}})} \right] \\
&= \lim_{t \to \infty}\frac{1}{t}\log \textbf{E}[ \mathrm e ^{\theta C_{ij}^{t-1}} ]\textbf{E}[ \mathrm e ^{-\theta M_{ij}^{t-1}} ] \notag \\
 & =  \lim_{t \to \infty} \frac{1}{t}\log \textbf{E}[ \mathrm e ^{\theta C_{ij}^{t-1}} ] + \lim_{t \to \infty} \frac{1}{t}\log \textbf{E}[ \mathrm e ^{-\theta M_{ij}^{t-1}} ] \notag \\
& =\psi _{C_{ij}^{t-1}}(\theta)+\psi _{M_{ij}^{t-1}}(-\theta).
\end{split}
\end{equation}
In the above equation, $\psi _{C_{ij}^{t-1}}(\cdot)$ and $\psi _{M_{ij}^{t-1}}(\cdot)$  are respectively the cumulant generating functions of $C_{ij}^{t-1}$ and $M_{ij}^{t-1}$.
 Because the cumulant generating function for any Poisson process $X$ can be calculated as $\log \textbf{E}[e^{\theta {X}}]=\textbf{E}{X}(e^\theta-1)$,
we have
\begin{equation}
\psi _{C_{ij}^{t-1}}(\theta)=\Lambda(e^\theta-1),
\end{equation}
and
\begin{equation}
\psi _{M_{ij}^{t-1}}(\theta)=\lambda(e^\theta-1).
\end{equation}
Thus, \eqref{psi} can be transferred to
\begin{equation}\label{npsi}
\psi ^*(x)=\sup_{\theta \in{ \mathbb{R}}}\left \{\theta {x} - \Lambda(e^\theta-1)-\lambda(e^{-\theta}-1)\right \}.
\end{equation}

 In the following, we will discuss how to solve \eqref{npsi} so as to obtain $I(b)$ in \eqref{loss}.  In fact, we only need to consider the case of $x>0$ in \eqref{npsi} because of $\frac{b}{t}>0$. Let $G(\theta)= \theta {x} - \Lambda(e^\theta-1)-\lambda(e^{-\theta}-1)$. For $\frac{\partial^{2} {G(\theta)}}{\partial {\theta^{2}}}=- \Lambda e^\theta-\lambda e^{-\theta}<0$,  $\theta^*$ that maximizes  $G(\theta)$ can be calculated with the condition $\frac{\partial {G(\theta)}}{\partial {\theta}}=0$. So,
 \begin{equation}
\theta^*=\log \frac{x+\sqrt{x^2+4\lambda\Lambda}}{2\Lambda}.
\end{equation}

Hence, we have
\begin{align}\label{eq:onehand111}
\psi^{*}(x)=& x\log \frac{x+\sqrt{x^2+4\lambda\Lambda}}{2\Lambda}-\frac{x+\sqrt{x^2+4\lambda\Lambda}}{2} \nonumber \\
   & - \frac{2\Lambda\lambda}{x+\sqrt{x^2+4\Lambda\lambda}}+\lambda+\Lambda.
 \end{align}
Let $x=b/t$ in \eqref{eq:onehand111}, we can obtain $\inf_{t \geq 2}t\psi ^*(\frac{b}{t})$ at $t=b/(\lambda-\Lambda)$.

Thus, we have
\begin{equation}
I(b)=b\log{\frac{\lambda}{\Lambda}}.
\end{equation}
$I(b)$ is the rate function of $P(L_{ij}>L)$ which slumps   exponentially with the selection valve $L$. Hence, $I(b)$  indicates the decay speed of the voting failure rate.
The larger $I(b)$  is, the slower the growth rate of the voting failure rate. Otherwise, the growth rate of the voting failure rate will leap, suggesting that
candidate $j$ should improve his availability as much as possible to increase his selection pressure so as to shorten $L_{ij}$ for decreasing $P(L_{ij}>L)$.

Fig. \ref{lab5}  illustrates the impacts of  $b$ and $\lambda-\Lambda$ on $I(b)$, i.e., the decay rate of the voting failure rate. It can be found that $I(b)$ increases as  either $b$ or $\lambda-\Lambda$  goes up. Considering that $L=lb$  and $\lambda-\Lambda$ is the expectation of $F_{ij}^t$, we can draw the first conclusion: no matter lifting the selection valve $L$  or enhancing the selection pressure of a candidate, the growth rate of his  voting failure rate can fall off. Furthermore,  it can be found that the decay rate $I(b)$ increases linearly with $b$ while creeps up logarithmically with  $\lambda-\Lambda$. Hence, our second conclusion is  the change of the voting failure rate is more sensitive to the selection valve rather than the selection pressure. Hence, to reduce the voting failure rate, the most effective way is to lift the selection valve so as to increase
the
vacancies of super nodes. However, more super nodes' mining in the system incurs more cost. Therefore, it is worth to study how to determine a proper selection valve to trade off the balance of the voting failure rate and the mining cost. To that aim, we introduce the following concept:
\begin{definition}[Effective selection valve]
\label{def:effective_b}
The effective selection valve $L^*(\epsilon)$ is the  minimum  selection valve that makes the voting failure rate of a candidate $j$ no more than a given threshold  $\epsilon \in (0,1]$, which is called the voting failure tolerance degree. That is
$$
L^*(\epsilon)=\min{\left\{L:P(L_{ij}>L)\leq \epsilon\right\}}.
$$
\end{definition}

Facing any candidate $j$ whose expected merit to voter $i$ and expected number of times being elected by voter $i$ are respectively $\lambda$ and $\Lambda$, meaning that his expected selection pressure is $\lambda-\Lambda$, the effective selection valve of voter $i$ can be calculated by taking advantage of the following theorem:
\newtheorem{thm}{\bf Theorem}[section]
\begin{thm}\label{thm2}
$
L^*(\epsilon)=-\log \epsilon/\log{\frac{\lambda}{\Lambda}},
$
given $\lambda>\Lambda>0$ and $\epsilon \in \left(0,1 \right]$.
\end{thm}

\begin{proof}
We have $\log {P(L_{ij}>L)}\sim {-lI(b)}$ when $l\rightarrow \infty$. In light of $P(L_{ij}>L)\leq \epsilon$, we have $I(b)\geq -\frac{\log \epsilon}{l}$. Combined with $I(b)=b\log{\frac{\lambda}{\Lambda}}$, we can derive the value of $L^*(\epsilon)$.
\end{proof}

Fig. \ref{lab6} illustrates that the effective selection valve $L^*(\epsilon)$  decreases logarithmically with the increase of the voting failure tolerance degree  $\epsilon$, implying that the reduced voting failure rate requirement makes voter $i$ have more vacancies of super nodes.  When $\epsilon$ is fixed, as the expectation of the selection pressure $\lambda-\Lambda$ rises up,  $L^*(\epsilon)$   decreases  logarithmically, which also can make sure there are enough qualified candidates to be super nodes. The selection valve $L$ is defined as the
ranking queue length of a candidate whose selection pressure
is ranked $K^{th}$, where $K$ is the number of super nodes that voters should elect. Therefore, there exists an optimal value of $K$, namely $K^*(\epsilon)$ corresponding to $L^*(\epsilon)$,  which can be actually used to guarantee the voting failure rate of a candidate is not higher than  the voting failure tolerance degree.  So $K^*(\epsilon)$  is a
valuable parameter for the voting consensus algorithm designer to guarantee the requirement of the voting failure rate of candidates, essentially ensuring the fairness of the Blockchain system.

However, when $K$ is fixed, i.e.,  the vacancies of super nodes cannot be changed,  a direct way for candidate $j$  to reduce his voting failure rate is to
improve his competitiveness, i.e., the expectation of his merit $\lambda$.
 So we introduce the following definition:
\begin{definition}[Effective expectation of  merit]
\label{def:effective_merit}
The effective expectation of  merit  $\lambda^*(\epsilon)$  is the  minimum merit  expectation of a candidate that makes his voting failure rate no more than the voting failure tolerance degree  $\epsilon \in (0,1]$. That is,
$$
\lambda^*(\epsilon)=\min{\left\{\lambda:P(L_{ij}>L)\leq \epsilon\right\}}.
$$
\end{definition}

Taking advantage of the similar method, we can prove the following theorem:

\begin{thm}\label{thm3}
Given $\Lambda>0$ and $\epsilon \in \left(0,1 \right]$, $$\lambda^*(\epsilon)=\Lambda e^{-\frac{\log \epsilon}{L}}.$$
\end{thm}
Fig. \ref{lab7} shows given the fixed $\Lambda$, the increase of either the  voting failure tolerance degree  $\epsilon$ or the selection valve $L$ will lead to the drop of  the effective expectation of  merit $\lambda^*(\epsilon)$, which will remain stable after  declining  to a certain level. The reason behind this fact is   that the merit of a candidate  includes two parts:  his availability and his profit brought to the voter. Even if the availability can be reduced responding to the reduced requirement ($\epsilon$) and the increased vacancies of super nodes ($L$), the profit that the candidate can bring to the voter cannot be decreased due to his fixed expectation number of times being elected (i.e., $\Lambda$). Thus, the effective expectation of  merit $\lambda^*(\epsilon)$ will not reduce to zero but stabilize to the expectation of profit brought to the voter.

\begin{figure}[t]
\centering
\subfigure[]{
\begin{minipage}{0.4\linewidth}
\includegraphics[width=4cm]{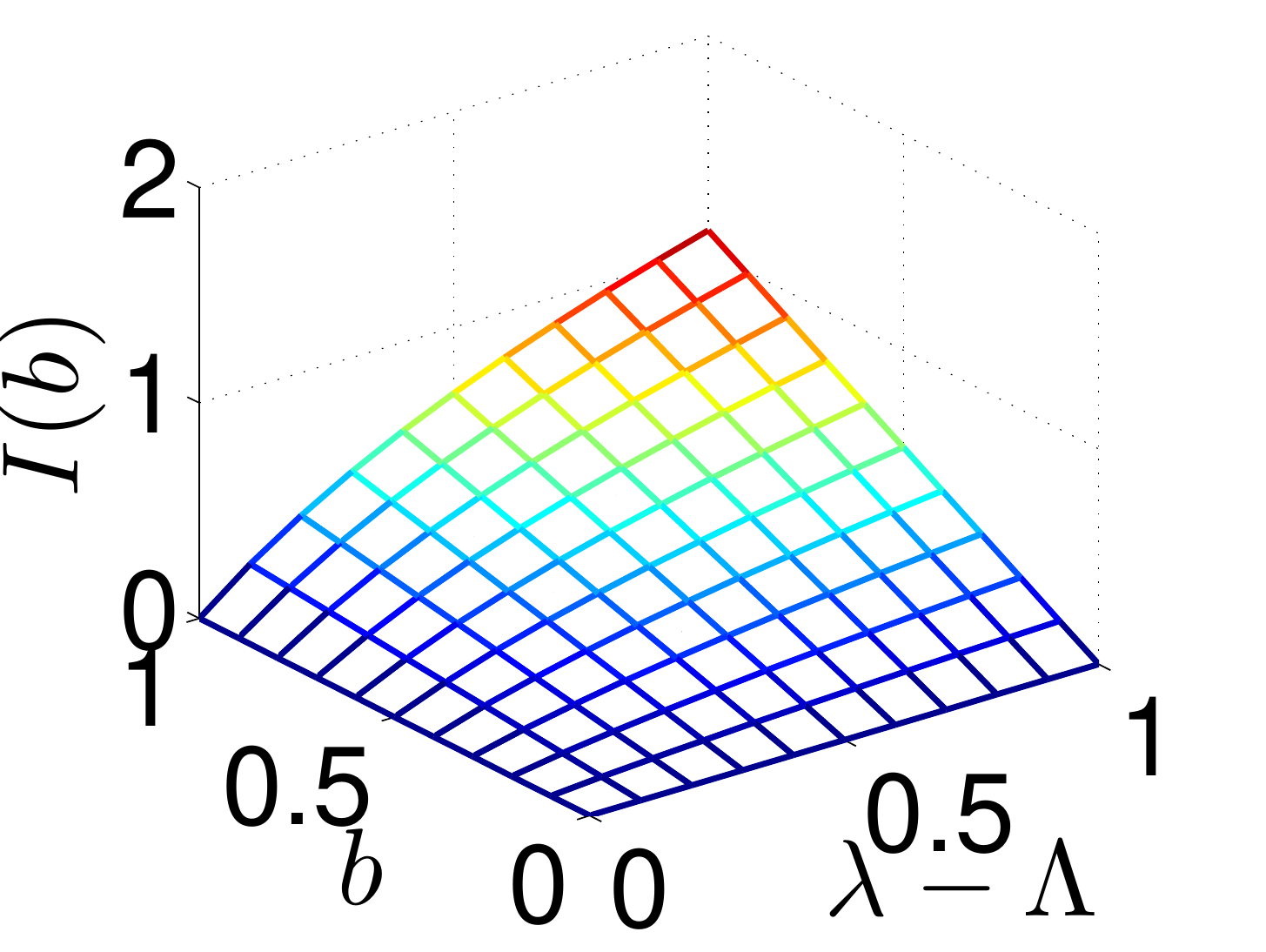}
\end{minipage}
}
\subfigure[]{
\begin{minipage}{0.22\linewidth}
\includegraphics[width=2cm]{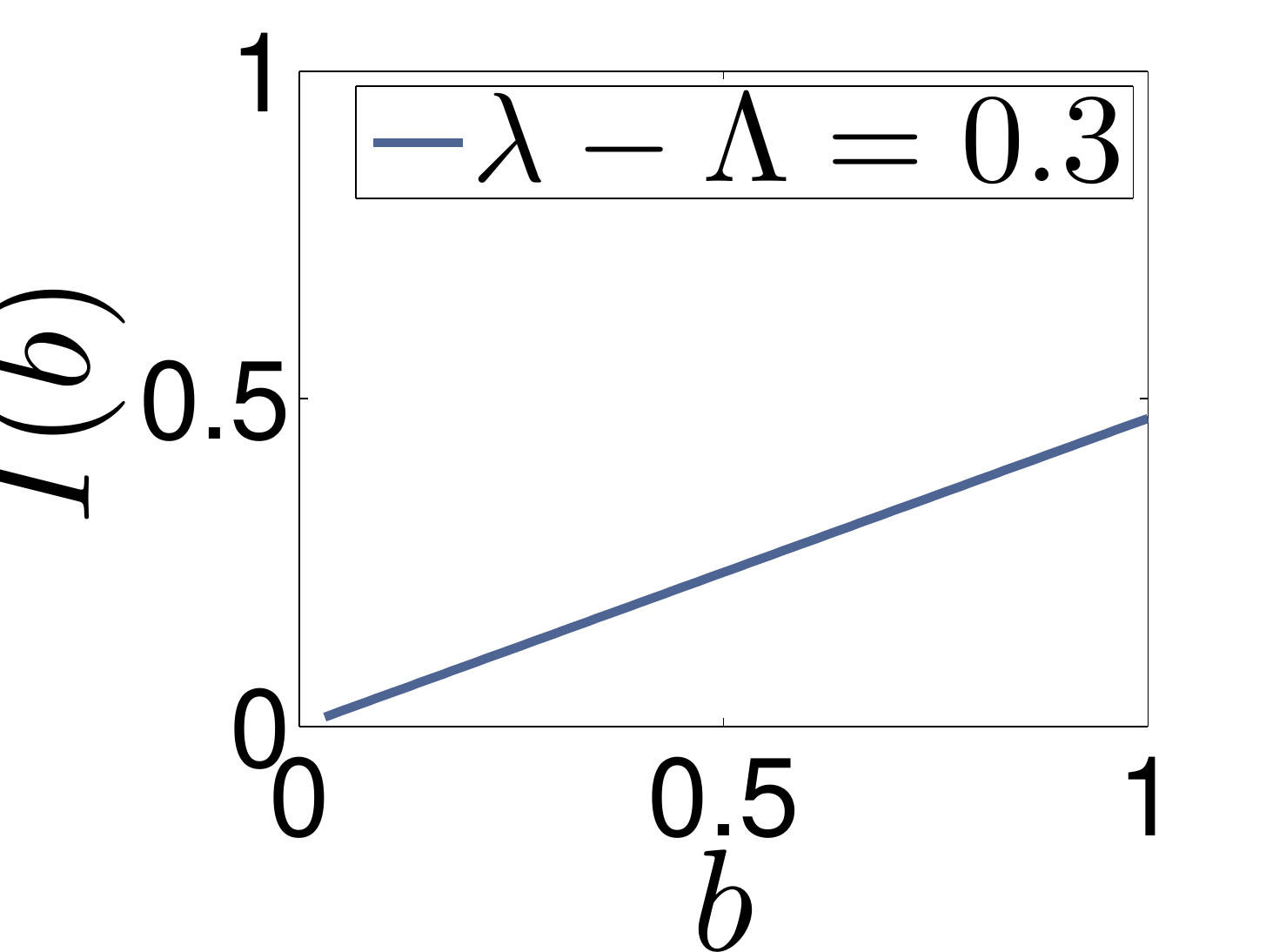}\\
\includegraphics[width=2cm]{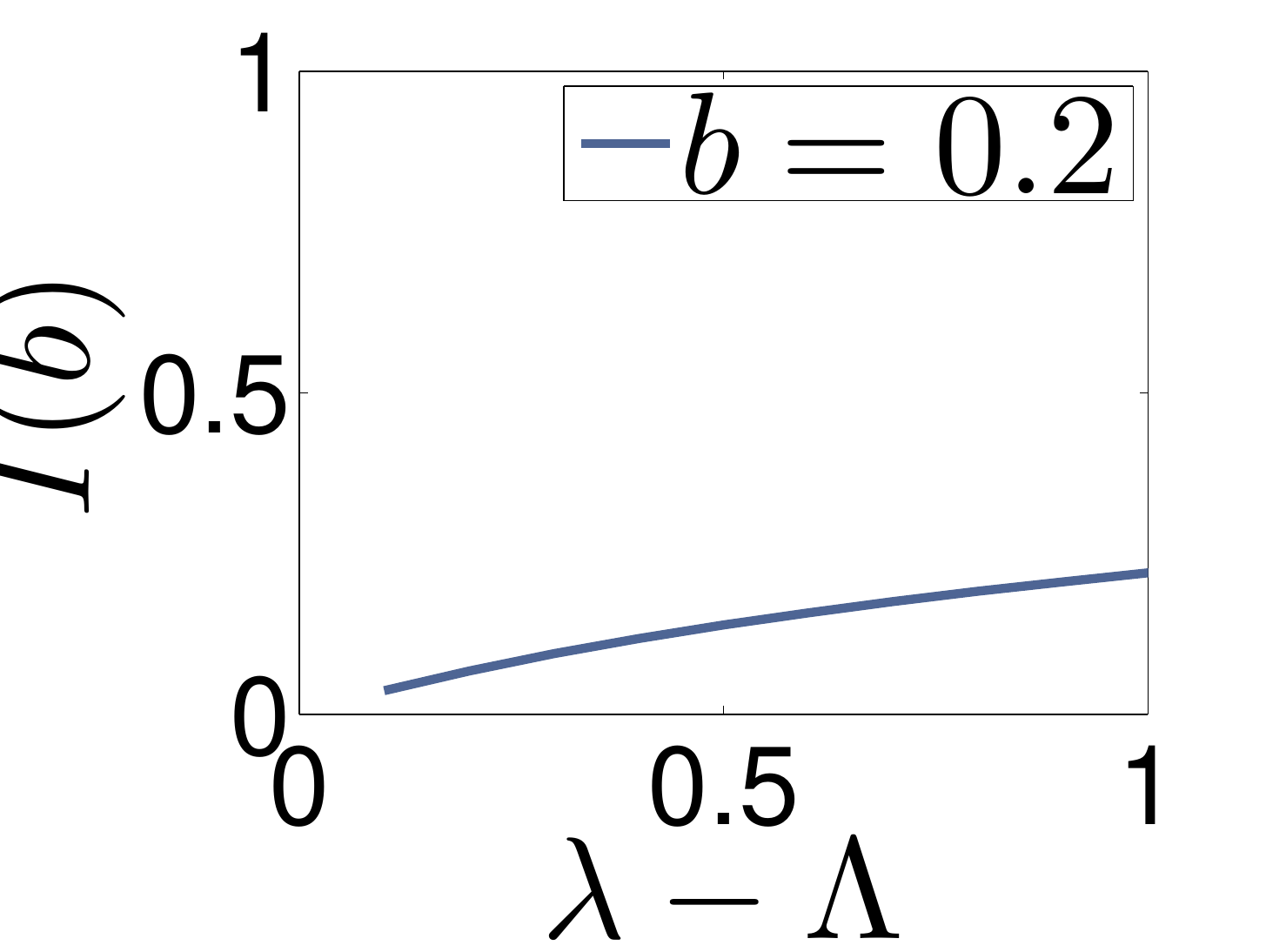}
\end{minipage}

\begin{minipage}{0.22\linewidth}
\includegraphics[width=2cm]{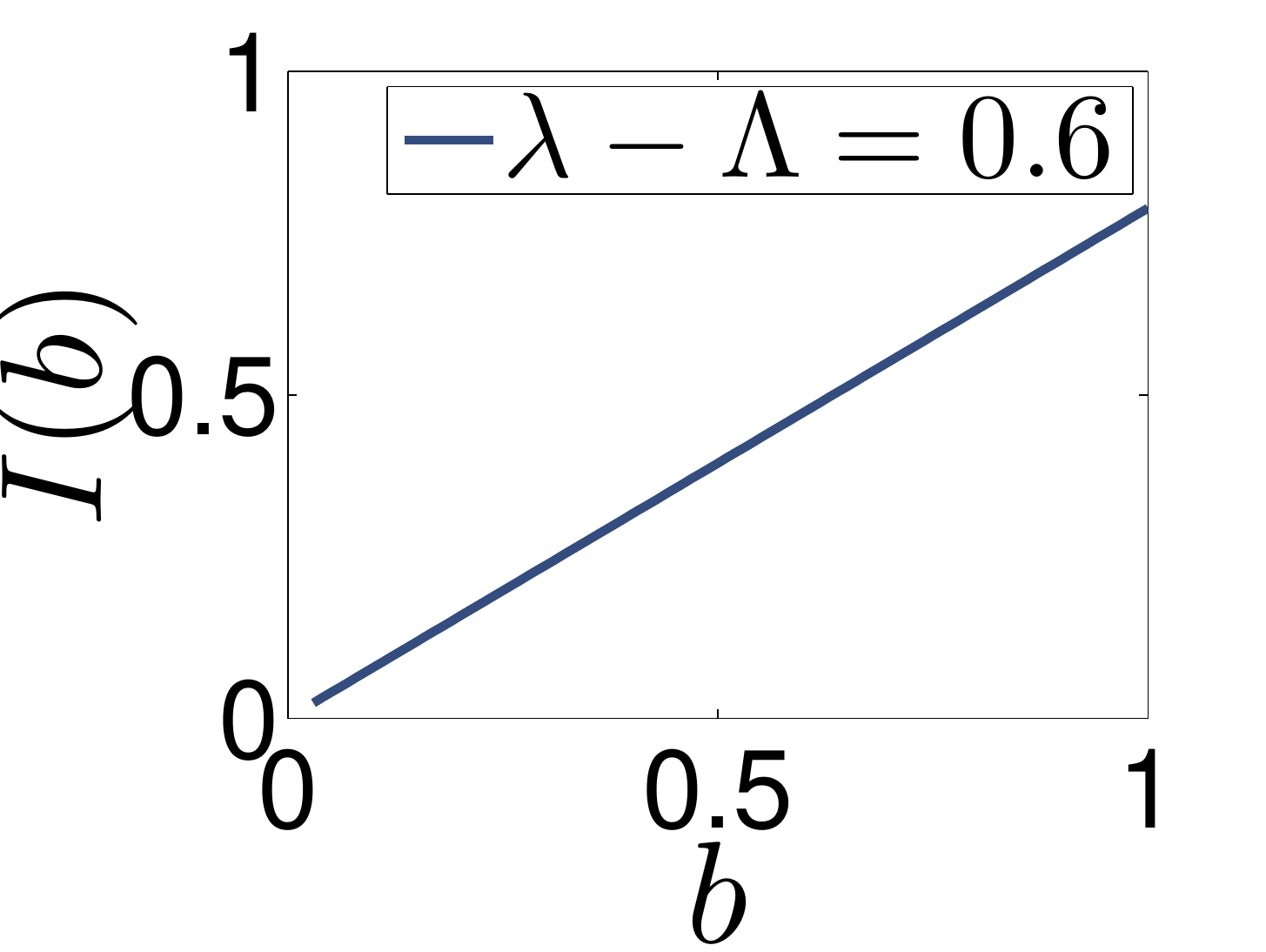}\\
\includegraphics[width=2cm]{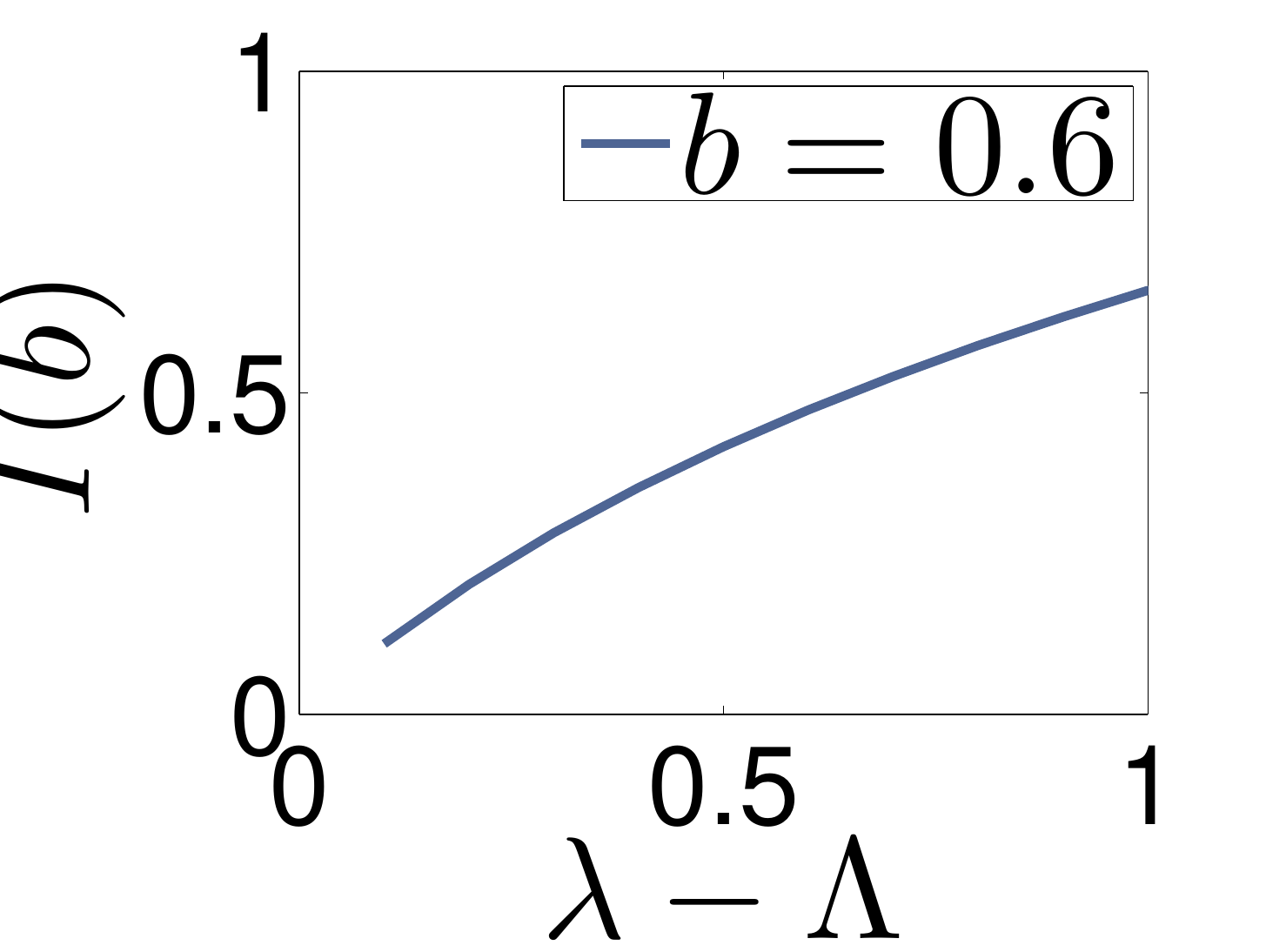}
\end{minipage}
}
\centering
\caption{Decay rate of the voting failure rate.}
\label{lab5}
\end{figure}

\begin{figure}[t]
\centering
\subfigure[]{
\begin{minipage}{0.4\linewidth}
\includegraphics[width=4cm]{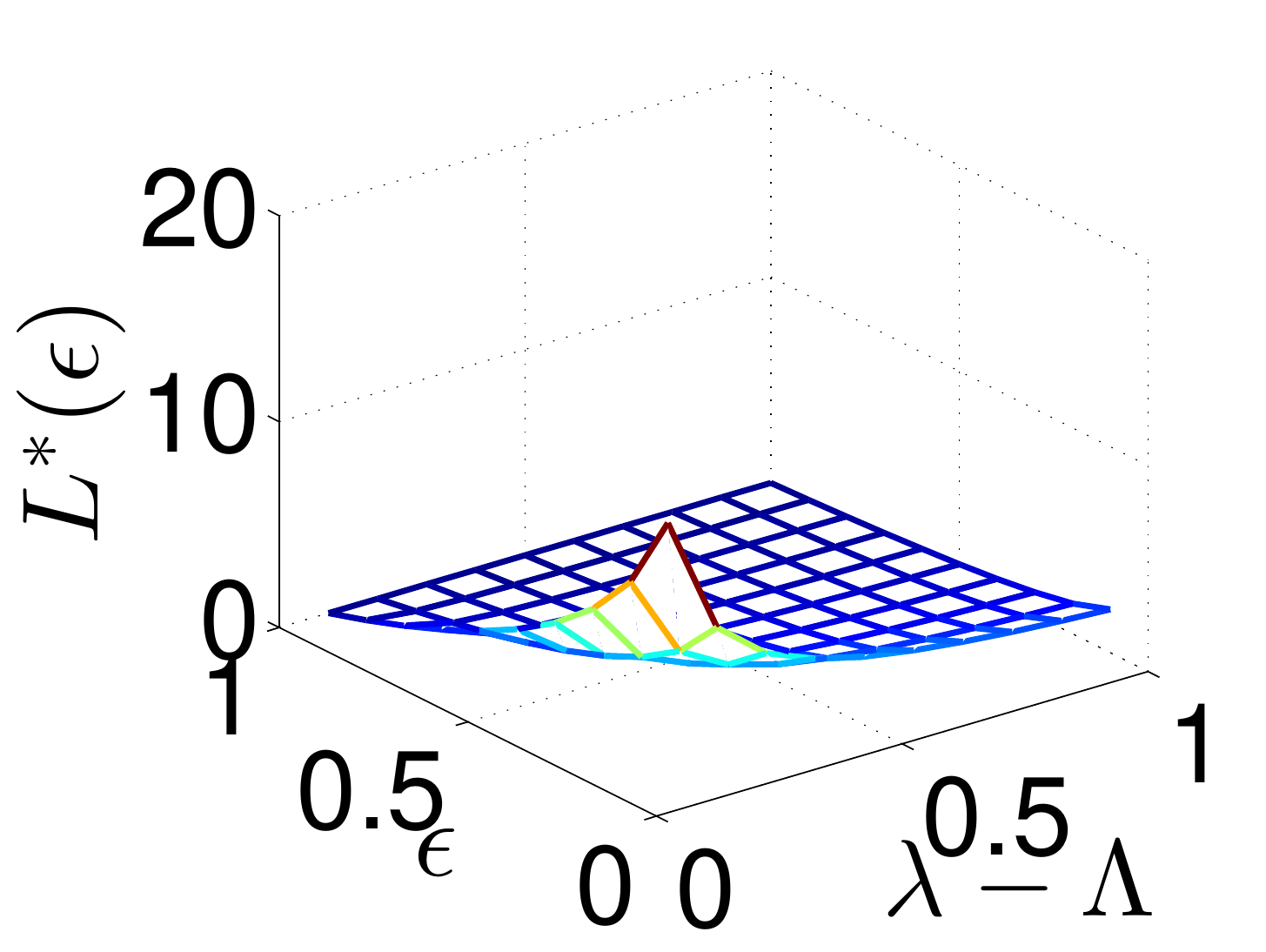}
\end{minipage}
}
\subfigure[]{
\begin{minipage}{0.22\linewidth}
\includegraphics[width=2cm]{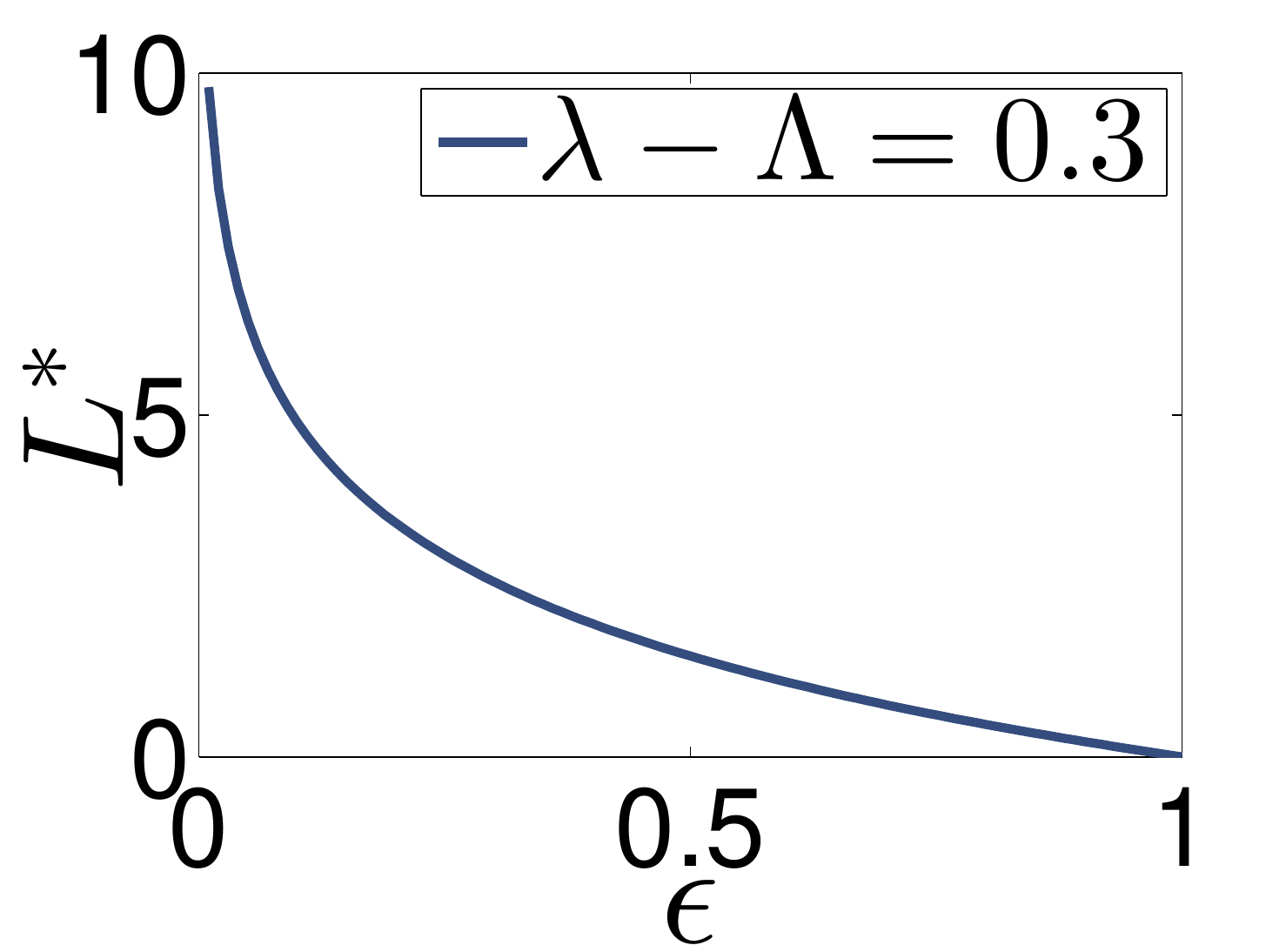}\\
\includegraphics[width=2cm]{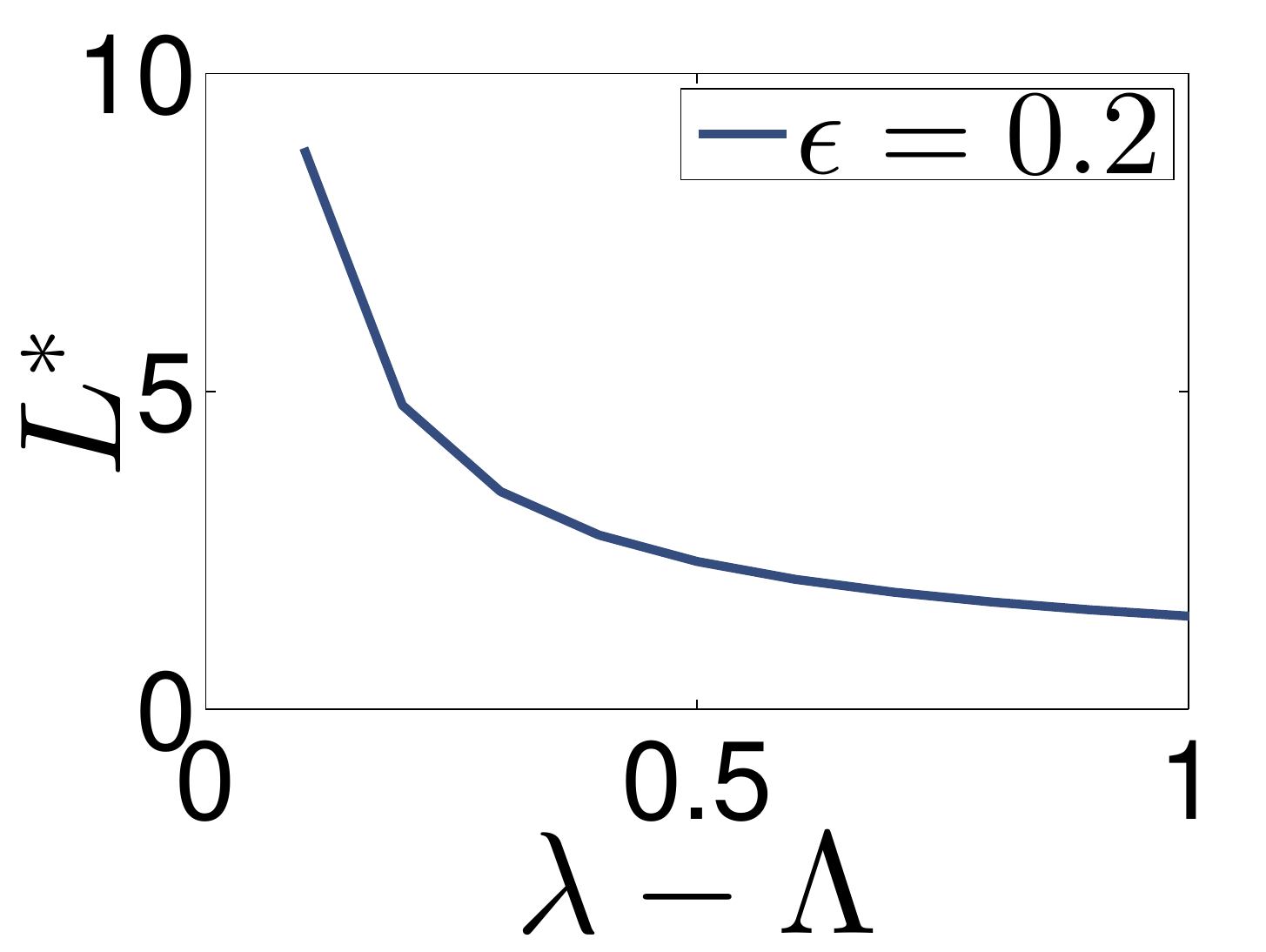}
\end{minipage}

\begin{minipage}{0.22\linewidth}
\includegraphics[width=2cm]{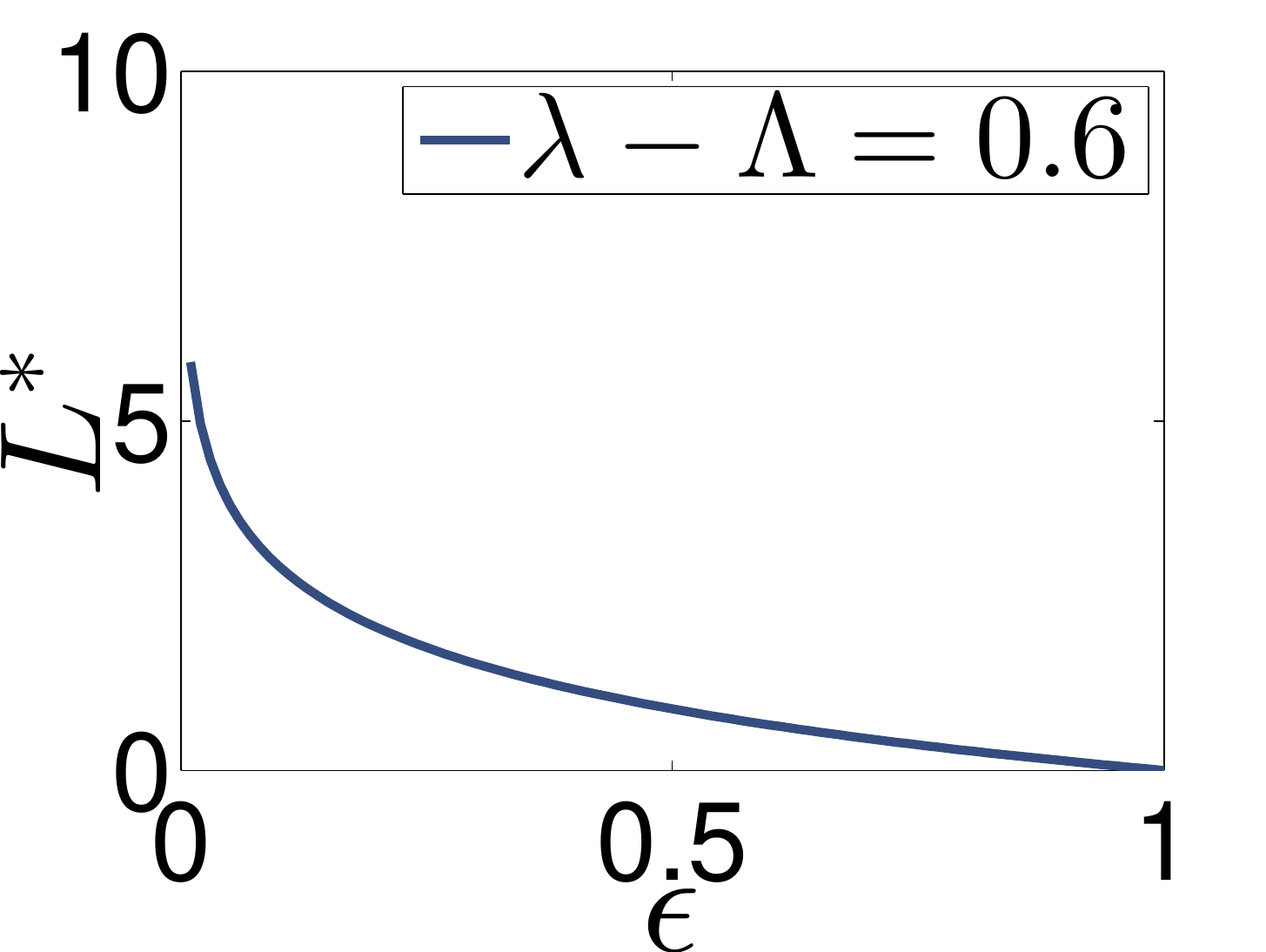}\\
\includegraphics[width=2cm]{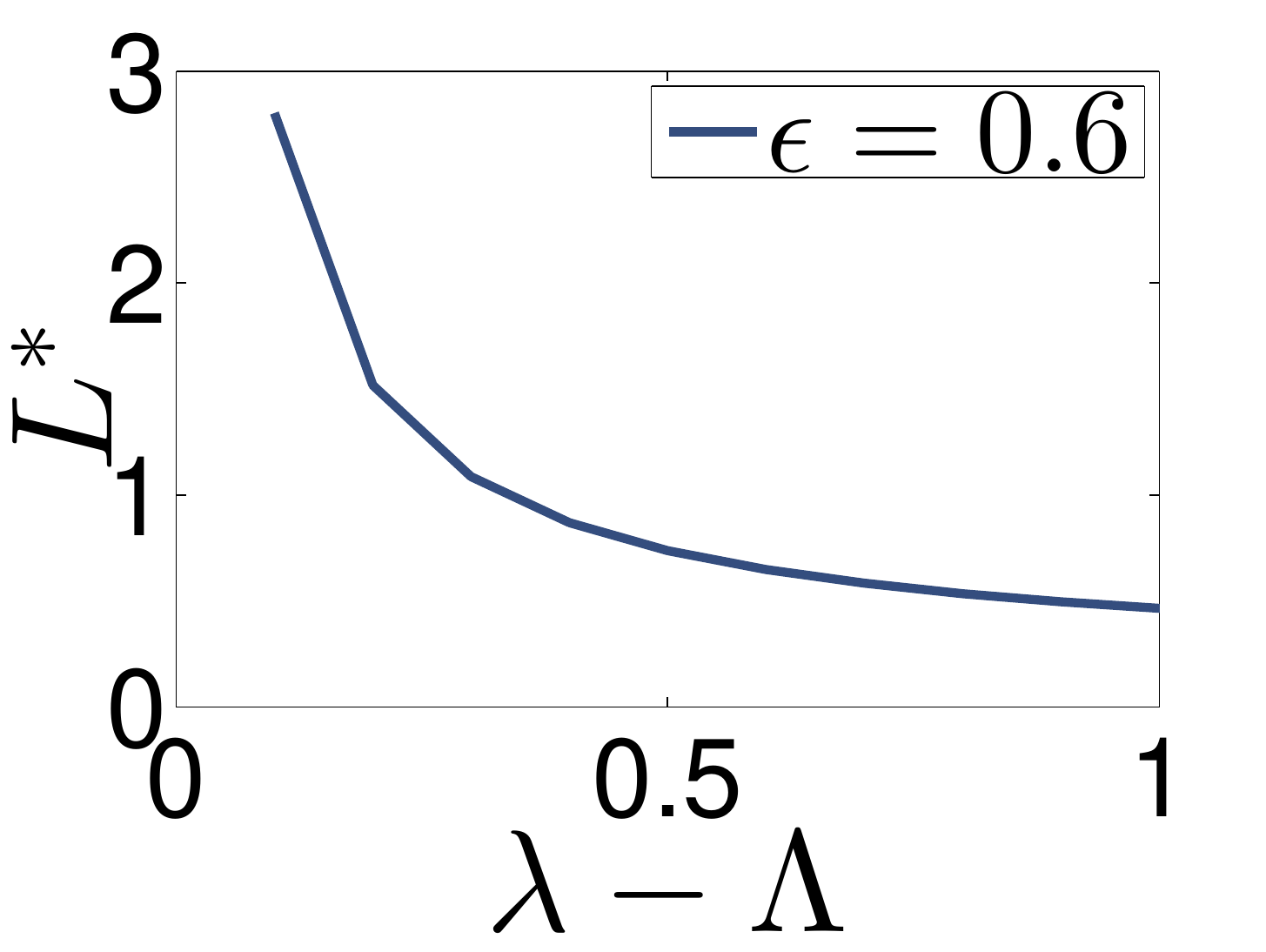}
\end{minipage}
}
\centering
\caption{Effective selection valve.}
\label{lab6}
\end{figure}

\begin{figure}[t]
\centering
\subfigure[]{
\begin{minipage}{0.4\linewidth}
\includegraphics[width=4cm]{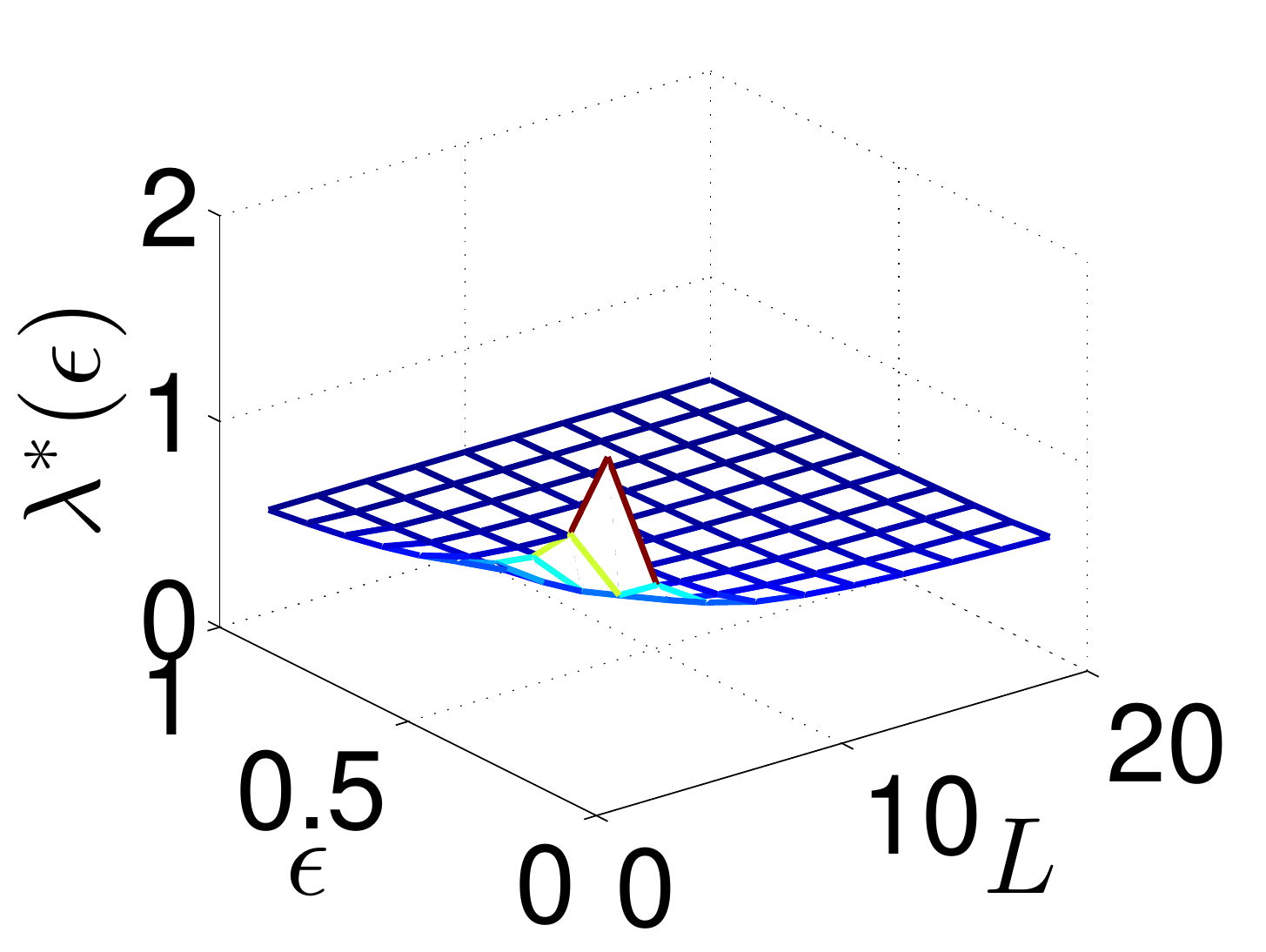}
\end{minipage}
}
\subfigure[]{
\begin{minipage}{0.22\linewidth}
\includegraphics[width=2cm]{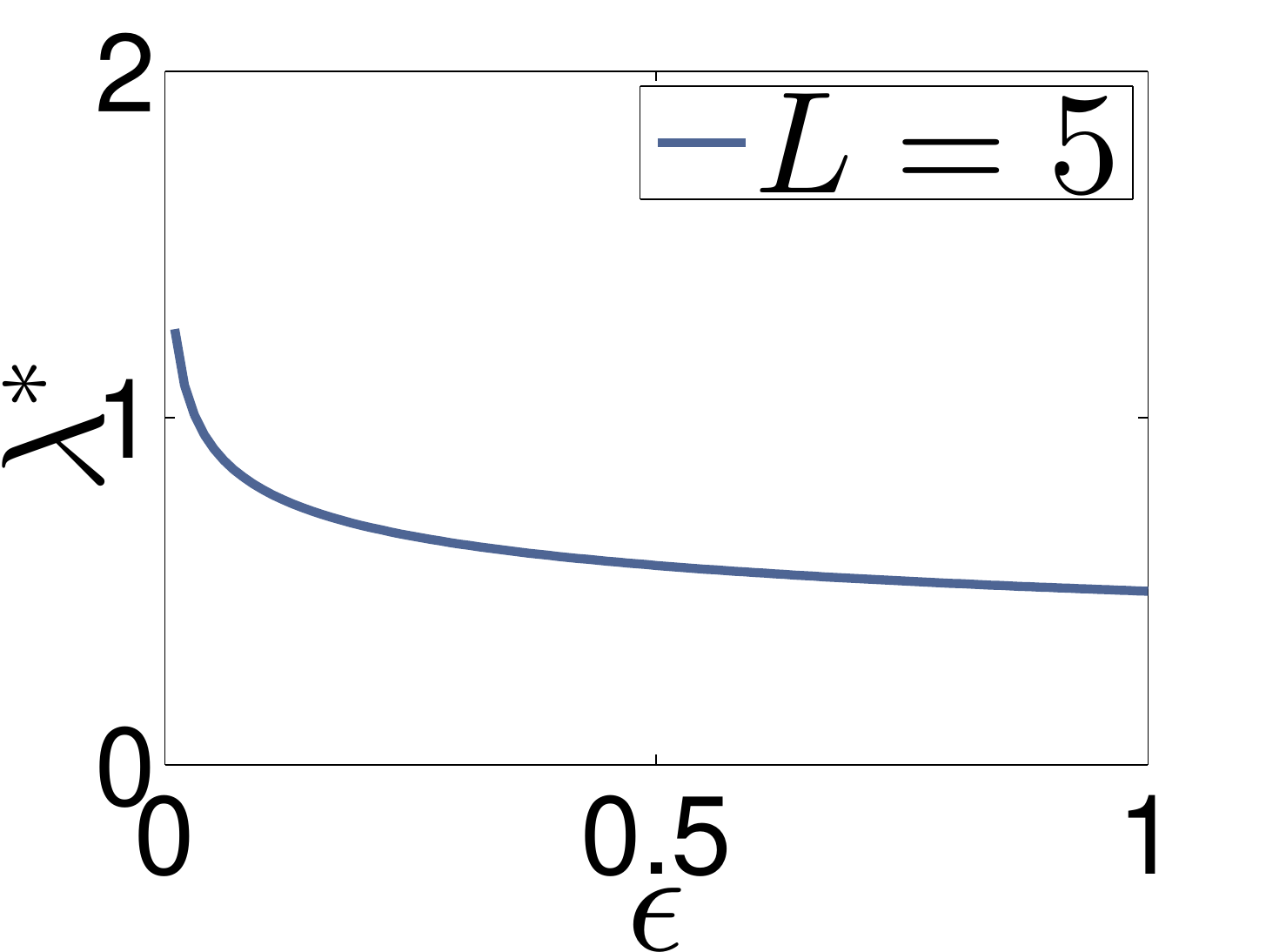}\\
\includegraphics[width=2cm]{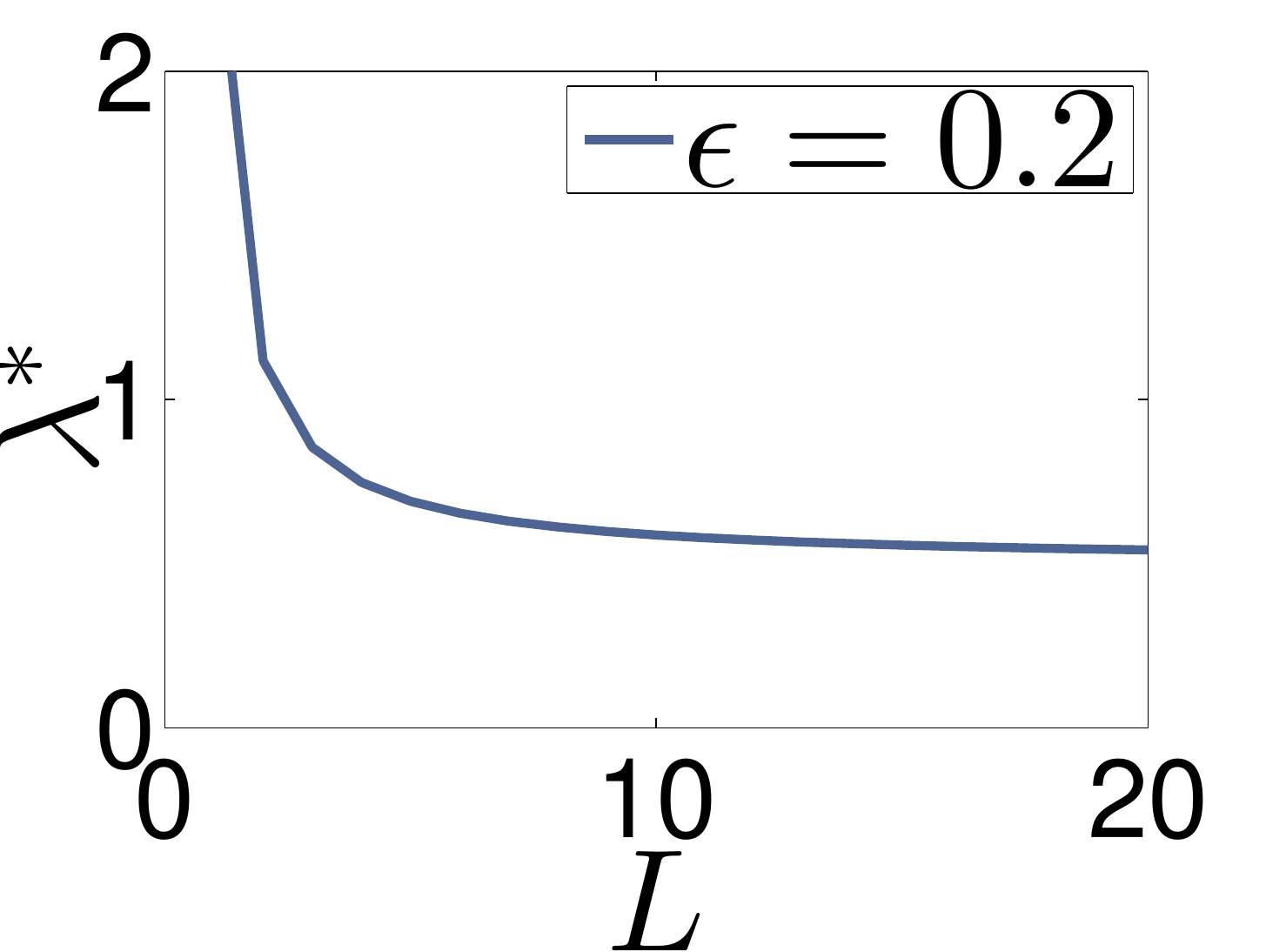}
\end{minipage}

\begin{minipage}{0.22\linewidth}
\includegraphics[width=2cm]{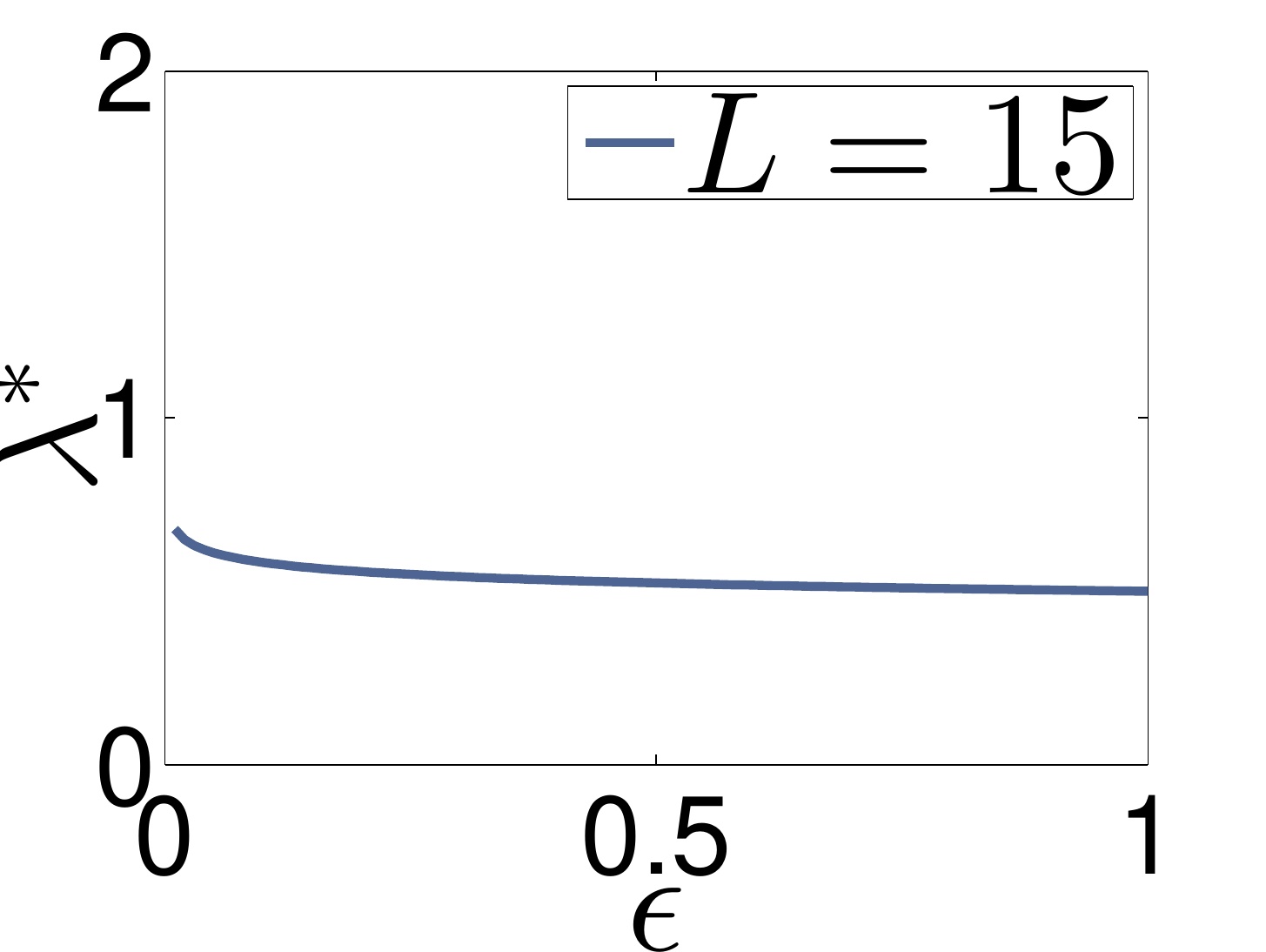}\\
\includegraphics[width=2cm]{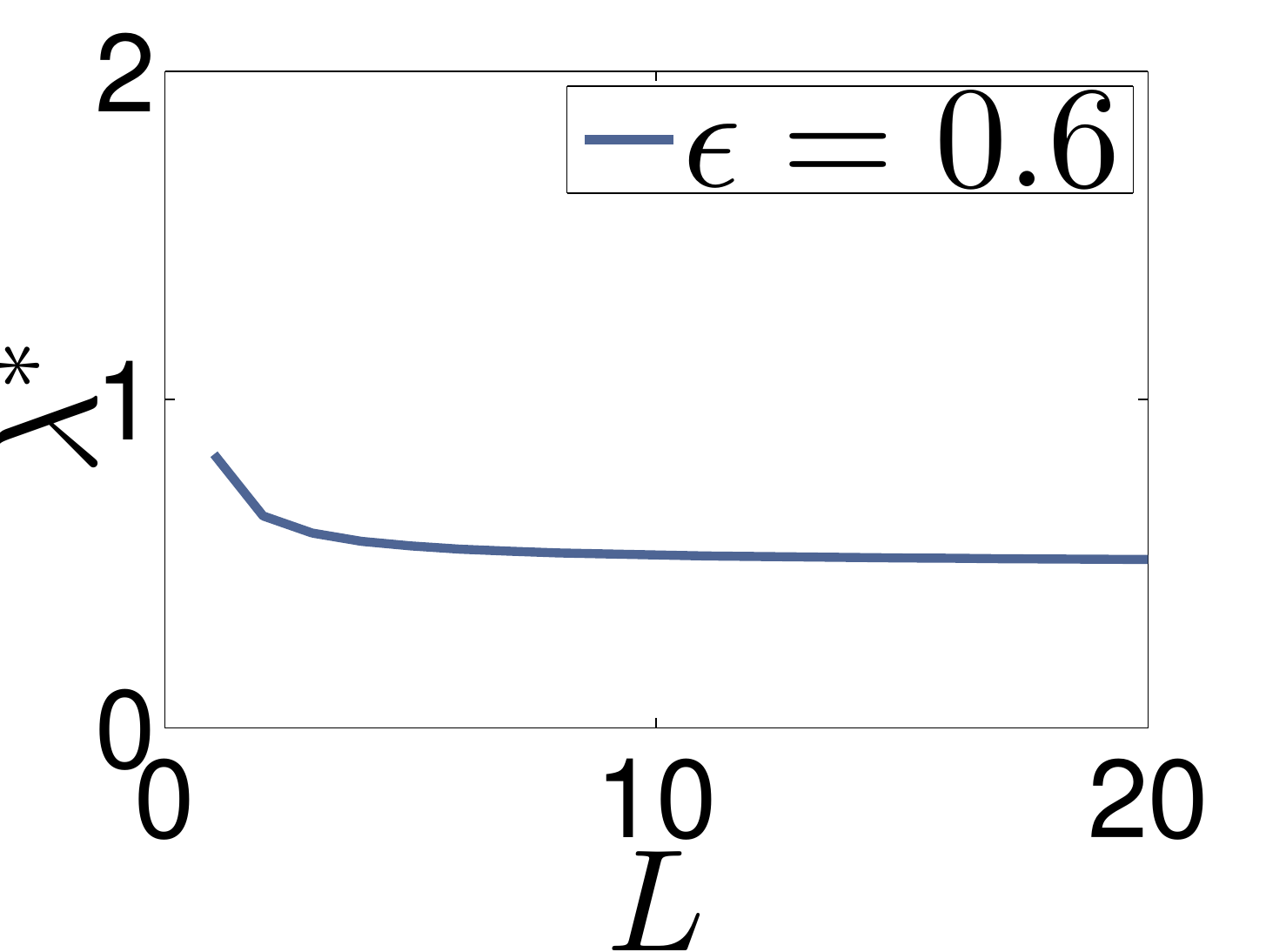}
\end{minipage}
}
\centering
\caption{Effective expectation of merit with $\Lambda=0.5$.}
\label{lab7}
\end{figure}

\section{Our Voting  Consensus Framework with Voting  Trustworthiness Evaluation}\label{sec:trustworthiness}

 In this paper, we employ  the idea of  the private-prior peer prediction \cite{Du18} \cite{miller2005eliciting} to evaluate the trustworthiness of votes that voter $i$ casts to any candidate $j$ so as to determine $t_{ij}^k$  ($i=1,2,\cdots,N; \  j=1,2,\cdots,M; \ k=1,2,...$) in \eqref{score}. Our  voting  consensus framework with trustworthiness evaluation is shown in Fig. \ref{flowchart},  which is detailed in the following.

\begin{figure}[t]
\centerline{
\includegraphics[width=9.5cm]{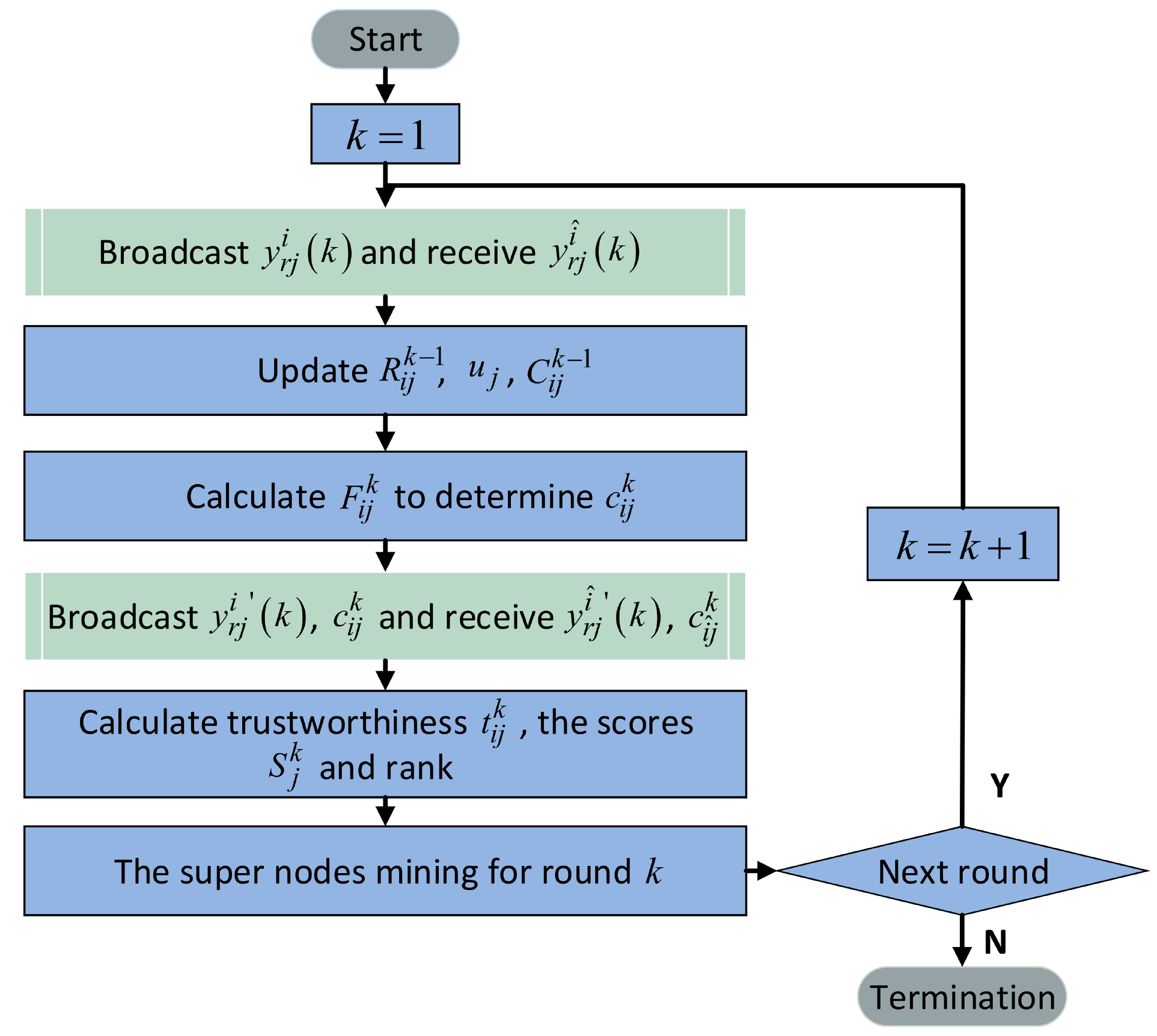}}
\caption{Our  voting  consensus framework with trustworthiness evaluation.}
\label{flowchart}
\end{figure}

In the beginning of round $k$,  any voter $i$ $ (i=1,2,\cdots ,N)$ is required to broadcast her prior belief ${y_{rj}^{i}}(k)\in{\left[0,1\right]}$  about  how likely her random peer voter $r$ $ (r=1,2,\cdots,i-1,i+1,\cdots,N)$ will vote for  candidate $j $ $(j=1,2,\cdots,M)$. At the same time, voter $i$ will also receive the prior belief ${y_{rj}^{\hat{i}}}(k)\in{\left[0,1\right]}$ from all other nodes in  the Blockchain system, where  $\hat{i}\neq i$ is also a voter.  When $k=1$, voters have to exchange the estimation of prior beliefs on selections without any prior knowledge. When  $k>1$, the prior belief $y_{rj}^{i}(k)$ can be calculated as:
\begin{equation}\label{pro}
\begin{split}
y_{rj}^{i}(k)=P(\hat{c_{rj}}=1|\hat{c_{ij}}=1)P(c_{ij}^k=1)\\
+P(\hat{c_{rj}}=1|\hat{c_{ij}}=0)P(c_{ij}^k=0).
\end{split}
\end{equation}
In \eqref{pro},  $P(c_{ij}^k=1)$ and $P(c_{ij}^k=0)$ are respectively the probabilities that voter $i$ considers herself may select or not select candidate $j$ in round $k$,  which can be obtained from her subjective prior belief on selecting this candidate. $\hat{c_{ij}}$ ($\hat{c_{rj}}$) indicates whether voter $i$ ($r$) picks candidate $j$ in the long term.    Hence,  $P(\hat{c_{rj}}=1|\hat{c_{ij}}=1)$ and $P(\hat{c_{rj}}=1|\hat{c_{ij}}=0)$ are  the conditional probabilities  that voter $r$ selects candidate $j$  when voter $i$  makes the same and different decisions respectively, both of which  can be calculated according to the historic results of previous $k-1$ rounds.
 In detail, when $k>1$, we have $P(\hat{c_{rj}}=1|\hat{c_{ij}}=1)=\frac{P(\hat{c_{rj}}=1,\hat{c_{ij}}=1)}{P(\hat{c_{ij}}=1)}
=\sum_{v=1}^{k-1}c_{rj}^v\cdot c_{ij}^v/\sum_{v=1}^{k-1}c_{ij}^v$, and $P(\hat{c_{rj}}=1|\hat{c_{ij}}=0)=\frac{P(\hat{c_{rj}}=1,\hat{c_{ij}}=0)}{P(\hat{c_{ij}}=0)}
=\sum_{v=1}^{k-1}c_{rj}^v\cdot (1-c_{ij}^v)/\sum_{v=1}^{k-1}(1-c_{ij}^v)$.

Then, each voter $i$ updates: a) the  unavailable  times of each candidate $j$ before round $k$,  so as to deduce his unavailable probability $u_j$; b) the selection times $C_{ij}^{k-1}$ of candidate $j$ before $k$, based on which the profit $R_{ij}^{k-1}$  that candidate $j$ devoted to voter $i$  can be obtained; employing the above information, the selection pressure $F_{ij}^{k}$  of each candidate $j$ in round $k$ can be calculated according to \eqref{wij}, so that the voting choice $c_{ij}^k$ ($i=1,2,\cdots,N; \  j=1,2,\cdots,M$) can be obtained through the proposed selection pressure-based voting consensus algorithm.

 After updating and calculating the above key parameters, the cognition of voter $i$ on each candidate iterates, which makes voter $i$ have the  posterior belief ${y_{rj}^{i}}'(k)\in{\left[0,1\right]}$ about  the possibility that the random  peer voter $r$ $ (r=1,2,\cdots,i-1,i+1,\cdots,N)$   will vote for any candidate $j $ $(j=1,2,\cdots,M)$. Then, voter $i$ will notify his choice $c_{ij}^k$ as well as the  posterior belief  ${y_{rj}^{i}}'(k)$  
 to all other nodes, and receive these two parameters from each of them.

  We use $V_j \in \left\{h,l \right\}$ to indicate the mining capability of candidate $j$.  If candidate $j$   is qualified enough to be elected as a super node,  $V_j=h$, and otherwise  $V_j=l$. Let $e_j^k \in \left\{0,1 \right\}$ represent whether candidate $j$ is elected as a super node in round $k$. Since the proposed  voting  consensus framework  can guarantee the voting results are credible, $P(V_j=h)$ can be estimated by $\frac{\sum_{v=1}^{k-1}e_j^v}{k-1}$ before the super nodes are elected in round $k$. Obviously, $P(V_j=l)=\frac{\sum_{v=1}^{k-1}(1-e_j^v)}{k-1}$ in this case. Thus,  the  posterior belief ${y_{rj}^{i}}'(k)$ when $k>  1$ can be calculated as\footnote{When $k=1$, the  posterior belief ${y_{rj}^{i}}'(1)$ is an estimated value  due to no knowledge on each candidate. }
\begin{eqnarray}
\label{eq:log1}
{y_{rj}^{i}}'(k)=
\begin {cases}
  P(\hat{c_{rj}}=1|V_j=h)P(e_j^k=1|c_{ij}^k=1)\\
+P(\hat{c_{rj}}=1|V_j=l)P(e_j^k=0|c_{ij}^k=1),\ \ \ c_{ij}^k=1 \nonumber \\
\\
P(\hat{c_{rj}}=1|V_j=h)P(e_j^k=1|c_{ij}^k=0)\\
+P(\hat{c_{rj}}=1|V_j=l)P(e_j^k=0|c_{ij}^k=0), \ \ \ c_{ij}^k=0
\end {cases}
\end{eqnarray}
where $P(\hat{c_{rj}}=1|V_j=h)=\frac{P(\hat{c_{rj}}=1,V_j=h)}{P(V_j=h)}
=\sum_{v=1}^{k-1}c_{rj}^v \cdot e_{j}^{v}/\sum_{v=1}^{k-1}e_{j}^{v}$ and $P(\hat{c_{rj}}=1|V_j=l)=\frac{P(\hat{c_{rj}}=1,V_j=l)}{P(V_j=l)}
=\sum_{v=1}^{k-1}c_{rj}^v \cdot (1-e_{j}^{v})/[(k-1)-\sum_{v=1}^{k-1}e_{j}^{v}]$.  $P(e_j^k=1|c_{ij}^k=1)$ and  $P(e_j^k=0|c_{ij}^k=1)$ ($P(e_j^k=1|c_{ij}^k=0)$ and $P(e_j^k=0|c_{ij}^k=0)$) respectively represent the  probabilities that candidate $j$ will be elected  or not in round $k$     when voter $i$  picks (does not choose) candidate $j$ in this round. Because in this stage, whether candidate $j$ is finally  elected as a super node is not determined yet,   $P(e_j^k=1|c_{ij}^k=1)$ and  $P(e_j^k=0|c_{ij}^k=1)$  ($P(e_j^k=1|c_{ij}^k=0)$ and $P(e_j^k=0|c_{ij}^k=0)$) need to be predicted by voter $i$, thus reflecting her real thoughts on whether candidate $j$ is qualified enough to be elected.

 Based on the prior and posterior beliefs, i.e.,  $y_{rj}^{i}(k)$ and  ${y_{rj}^{i}}'(k)$, as well as  the voting choice $c_{rj}^k$ received, the trustworthiness $t_{ij}^k$ of voter $i$'s choice $c_{ij}^k$ in round $k$ can be  calculated according to the following scoring rule \cite{a2}:
\begin{small}
\begin{equation}
\label{validity}
t_{ij}^k=\left [ \alpha W\left ( y_{rj}^{i}(k),c_{rj}^k \right )+\left ( 1-\alpha \right )W\left ( {y_{rj}^{i}}'(k),c_{rj}^k \right )+\beta \right ].
\end{equation}
\end{small}In \eqref{validity},
$\alpha\in [0,1]$ is a scaling parameter, and $W$ can be the logarithmic form, namely,
 \begin{equation}
\left\{\begin{matrix}
W\left ( y,c=1 \right )&=&\ln(y),\\
W\left ( y,c=0 \right )&=&\ln(1-y),
\end{matrix}\right.
\end{equation}
or the quadratic form, i.e.,
 \begin{equation}
\left\{\begin{matrix}
W\left ( y,c=1 \right )&=& 2y-{y^{2}},\\
W\left ( y,c=0 \right )&=& 1-{y^{2}}.
\end{matrix}\right.
\end{equation} And
\begin{small}
 \begin{equation}\label{b}
\beta=-\frac{1}{N}\sum_{i=1}^{N}\left [\alpha W\left ( y_{rj}^{i}(k),c_{rj}^k \right )+\left ( 1-\alpha \right )W\left ( {y_{rj}^{i}}'(k),c_{rj}^k \right )\right ].
\end{equation}
\end{small}
According to \eqref{b},  $\beta$ is used to measure how far the first two items in \eqref{validity} deviate from the average value.  Thus,   $t_{ij}^k>0$  represents trustworthy voting. The bigger the $t_{ij}^k$, the higher the credibility of the voting.  On the contrary,   $t_{ij}^k<0$ reveals an unreliable voting. The smaller the $t_{ij}^k$, the worse the voting credibility.

After calculating $t_{ij}^k$ and combining with $c_{ij}^k$ obtained in the last section, as well as  the stake of voter $i$, i.e., $s_i$, 
we can obtain the score $S_j^k$ of each candidate in round $k$  according to \eqref{score}.  The top $K$
nodes with the highest scores $S_j^k$ can be elected as super nodes in round $k$ finally.

\begin{theorem} \label{theorem:1}
The method for evaluating trustworthiness is incentive compatible.
\end{theorem}
\begin{proof}
As the expectation of $t_{ij}^k$, $\mathbb{E}(t_{ij}^k)$ can be calculated as
\begin{small}
\begin{align*}
\mathbb{E}[t_{ij}^k]=&E[\alpha W(y_{rj}^i(k),c_{rj}^k)]+E[(1-\alpha) W({y_{rj}^i}'(k),c_{rj}^k)]+E[\beta]\\
=&\alpha(1-\frac{1}{N})E[W(y_{rj}^i(k),c_{rj}^k)]\\
&+(1-\alpha)(1-\frac{1}{N})E[W({y_{rj}^i}'(k),c_{rj}^k)]\\
&-\frac{1}{N}\sum_{n=1,n \neq i}^{N} [\alpha W(y_{rj}^n(k),c_{rj}^k)+(1-\alpha)W({y_{rj}^n}'(k),c_{rj}^k)].
\end{align*}
\end{small}
Employing the logarithmic form of  $W$, we have
\begin{small}
\begin{align*}
\mathbb{E}[t_{ij}^k]&=\alpha(1-\frac{1}{N})[p_1\ln y_{rj}^i(k)+(1-p_1)\ln (1-y_{rj}^i(k))]\\
&+(1-\alpha)(1-\frac{1}{N})[p_2\ln {y_{rj}^i}'(k)+(1-p_2)\ln (1-{y_{rj}^i}'(k))]\\
&-\frac{1}{N}\sum_{n=1,n \neq i}^{N} [\alpha W(y_{rj}^n(k),c_{rj}^k)+(1-\alpha)W({y_{rj}^n}'(k),c_{rj}^k)].
\end{align*}
\end{small}

In light of
\begin{small}
\begin{align*}
\frac{\partial \mathbb{E}[t_{ij}^k]}{\partial y_{rj}^{i}(k)}&=\alpha(1-\frac{1}{N})\frac{p_1-y_{rj}^i(k)}{y_{rj}^i(k)(1-y_{rj}^i(k))}=0,\\
\frac{\partial \mathbb{E}[t_{ij}^k]}{\partial{y_{rj}^{i}}'(k)}&=\alpha(1-\frac{1}{N})\frac{p_2-{y_{rj}^{i}}'(k)}{{y_{rj}^{i}}'(k)(1-{y_{rj}^{i}}'(k))}=0,
\end{align*}
\end{small}we can deduce $y_{rj}^{i}(k)=p_1$ and ${y_{rj}^{i}}'(k)=p_2$  can satisfy the above equations, where $p_1=P(\hat{c_{rj}}=1)$,  $p_2=P(\hat{c_{rj}}=1|c_{ij}^k=0)$ when $c_{ij}^k=0$ and $=P(\hat{c_{rj}}=1|c_{ij}^k=1)$ when  $c_{ij}^k=1$.
Since
\begin{small}
\begin{align*}
\left.\frac{\partial \mathbb{E}^{2}[t_{ij}^k]}{\partial {y_{rj}^{i}}^{2}(k)}\right |_{y_{rj}^{i}(k)=p_1}
&=\alpha(1-\frac{1}{N})\frac{y_{rj}^{i}(k)(y_{rj}^{i}(k)-1)}{{y_{rj}^{i}}^2(k)(1-{y_{rj}^{i}}^2(k))}<0, \\
\left.\frac{\partial \mathbb{E}^{2}[t_{ij}^k]}{\partial {{y_{rj}^{i}}'}^{2}(k)}\right |_{{y_{rj}^{i}}'(k)=p_2}
&=\alpha(1-\frac{1}{N})\frac{{y_{rj}^{i}}'(k)({y_{rj}^{i}}'(k)-1)}{{{y_{rj}^{i}}'}^2(k)(1-{{y_{rj}^{i}}'}^2(k))}<0,
\end{align*}
\end{small}$y_{rj}^{i}(k)=p_1$ and ${y_{rj}^{i}}'(k)=p_2$ can maximize  $\mathbb{E}[t_{ij}^k]$.
Because $p_1$ and $p_2$ represent respectively the prior and posterior probabilities that an arbitrary voter $i$ supports candidate $j$, reflecting his objective mining capacity. Hence, $y_{rj}^{i}(k)=p_1$ and ${y_{rj}^{i}}'(k)=p_2$ indicate the subjective beliefs tally with the objective reality, implying the voter reports the truth in this case. In another word, only a voter behaves honestly, can she maximizes her trustworthiness of voting, meaning the method for evaluating trustworthiness is incentive compatible. When we use  the quadratic form to calculate $W$, the same conclusion can be drawn.
\end{proof}

\section{Experimental Evaluation}
\label{sec:experiment}
In this section, we analyze the impacts of some key parameters on the reward of voters and testify the effectiveness of our proposed voting consensus mechanism. We conduct a large number of simulations under different parameter settings, but only report partial results derived in a few parameter settings as follows, because other parameter settings present similar trends and thus we omit them to avoid redundancy.

Basically, we consider there are 50 candidate miners in total and set the reward of generating a new block as 12.5 which is in line with the current setting of Bitcoin and will be shared by all voters of the winning candidate $j$, i.e., $\sum_i R^k_{ij}=12.5$.
Also, we assume that voters own different numbers of stakes  indicating their different weights during the voting process, where any voter $i$ can choose $s_i \in \{ 1,2,3,4\}$. Besides, we define that all voters can give each candidate up to one vote, and each voter needs to choose the top $K$ super nodes that she considers reliable to give one vote.

\subsection{Availability Function}
\label{s_6_2}
First, in order to study the effect of unavailable probability on the voter's merit, we calculate the cumulative reward of any voter until round $t$ with different availability functions in Fig. \ref{lab2}, as well as different numbers of stake for the voter.
As shown in Fig. \ref{lab2}(a), we present the cumulative reward of the voter when her stake varies and the availability function changes, as well as the case of not considering unavailability factor, till the round $t=100$.
While Fig. \ref{lab2}(b) illustrates the three specific availability functions, i.e., $d_1$, $d_2$, and $d_3$,
varying with unavailable probability $u_j \in \left [ 0,1\right ]$. The three types of functions represent power function, exponential function and linear function respectively. And the scaling parameter in \eqref{aijt} is $\rho=5$, which satisfies the requirement to guarantee the stability of virtual queue \footnote{Other values that meet the requirements can be implemented in a similar way.}.

In the light of Fig. \ref{lab2}(a), we can observe that with the increase of the voter's stake, the cumulative reward increases correspondingly. Further more, through comparing the cumulative reward in the case of considering the unavailability and without it, one can find that including the unavailability is significant to affect the reward of the voter, while when it comes to the cumulative reward with three different availability functions, varying the decline speed of availability function has no obvious effect on the voter's profit. Also, with the increasing number of stake, the influence of whether involving unavailability factor or not becomes more significant, which can be seen from the increasingly large gap between the results of considering no unavailability and including it.

\begin{figure}[t]
\centering
\subfigure[Reward comparison.]{
\begin{minipage}[t]{0.4\linewidth}
\centering
\includegraphics[width=1.1\textwidth]{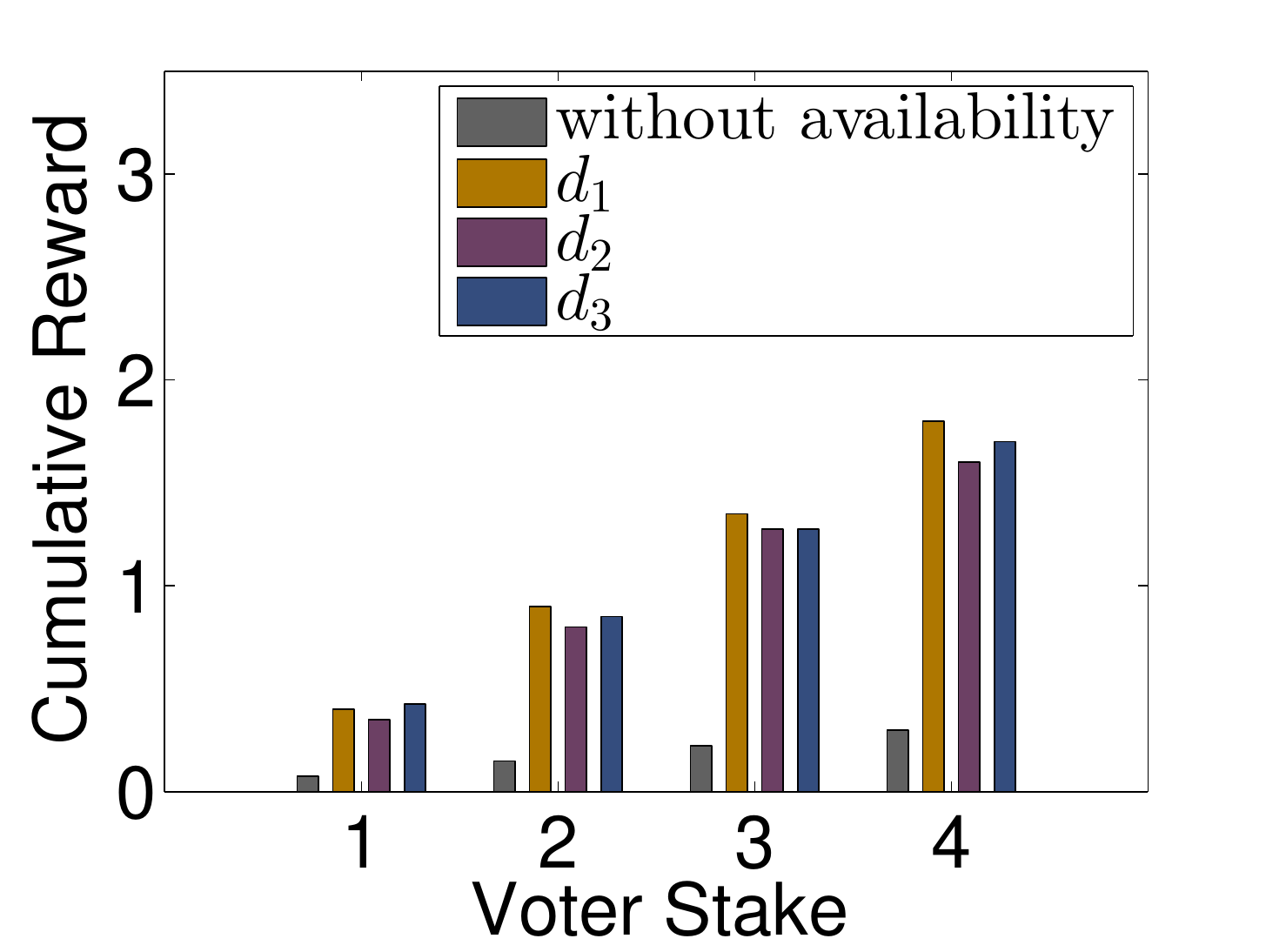}
\end{minipage}
}
\subfigure[Availability functions.]{
\begin{minipage}[t]{0.4\linewidth}
\centering
\includegraphics[width=1.1\textwidth]{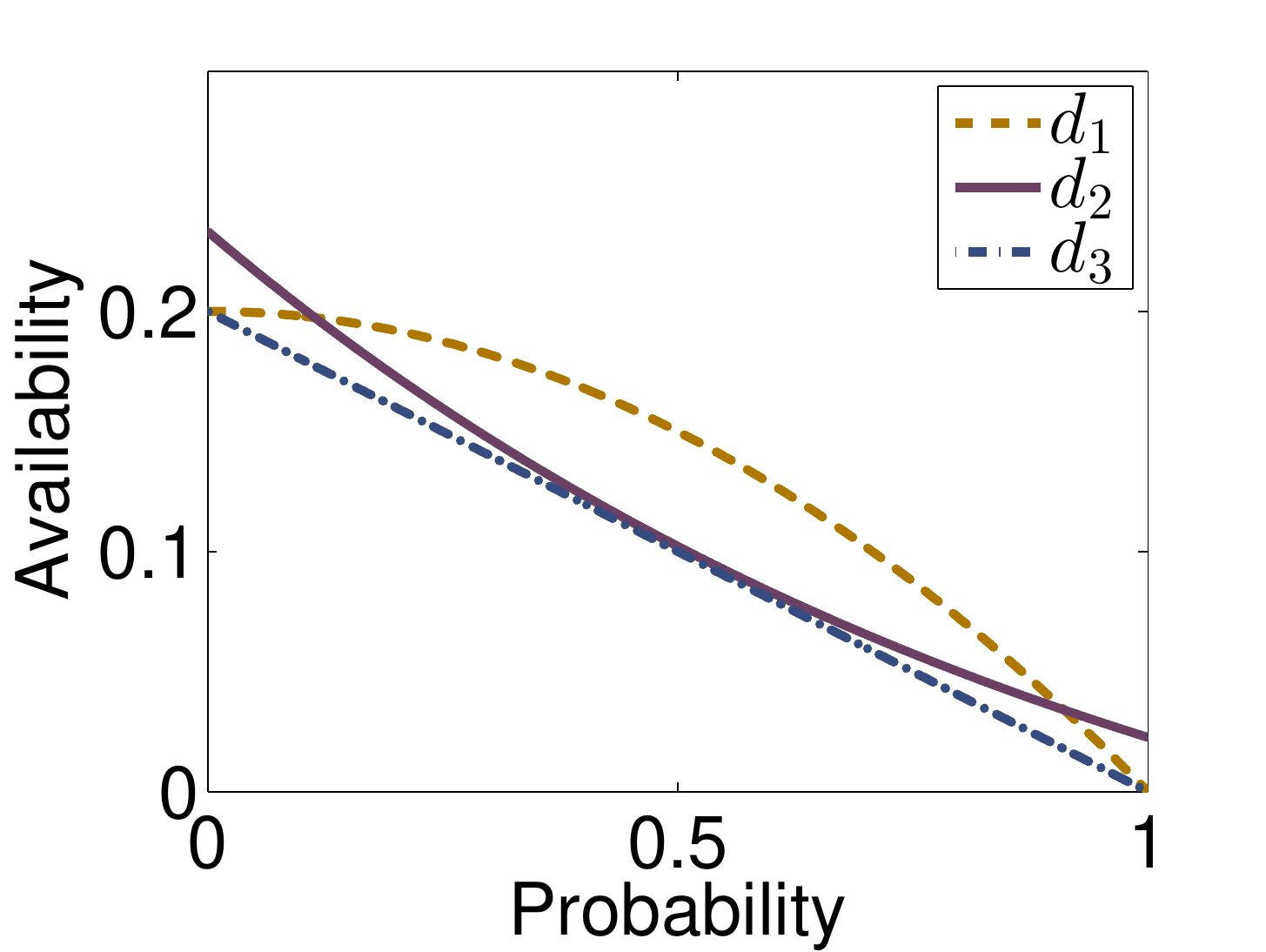}
\end{minipage}
}
\centering
\captionsetup{font={small}}
\caption{Cumulative reward of the voter with different availability functions and numbers of stake.}
\label{lab2}
\end{figure}

\subsection{Evolution of the Reward}
\label{s_6_3}
Next, we investigate the evolution of the voter's reward over time under different numbers of stake and different availability functions, where the experimental results are presented in Figs. \ref{lab3} and \ref{lab4}, respectively.

According to Fig. \ref{lab3}, one can easily observe that the more stake the voter voting to the winning candidates, the higher her cumulative reward is, which is consistent with the general expectation. In detail, the stake of voter $i$, i.e., $s_i$, is used as the voting weight, which affects her reward obtained from the winning candidate. Considering the case where the voter votes the same candidates as super nodes, she can receive reward from the successfully elected candidates she supported, which is positively related to her voting weights as the super nodes divide reward according to the proportion of their supporters' votes, i.e., $\frac{s_i\times t^k_{ij}\times c^k_{ij}}{\sum_{k=1}^{N}s_i\times t^k_{ij}\times c^k_{kj}}=\frac{s_i\times t^k_{ij}\times c^k_{ij}}{S^k_j}$. Thus, the greater the weight, the more the reward. Besides, because the voter will never suffer from any economic loss during the voting process, her total reward is increasing as the number of rounds increases.

\begin{figure}[t]
\centering
\includegraphics[width=0.4\linewidth]{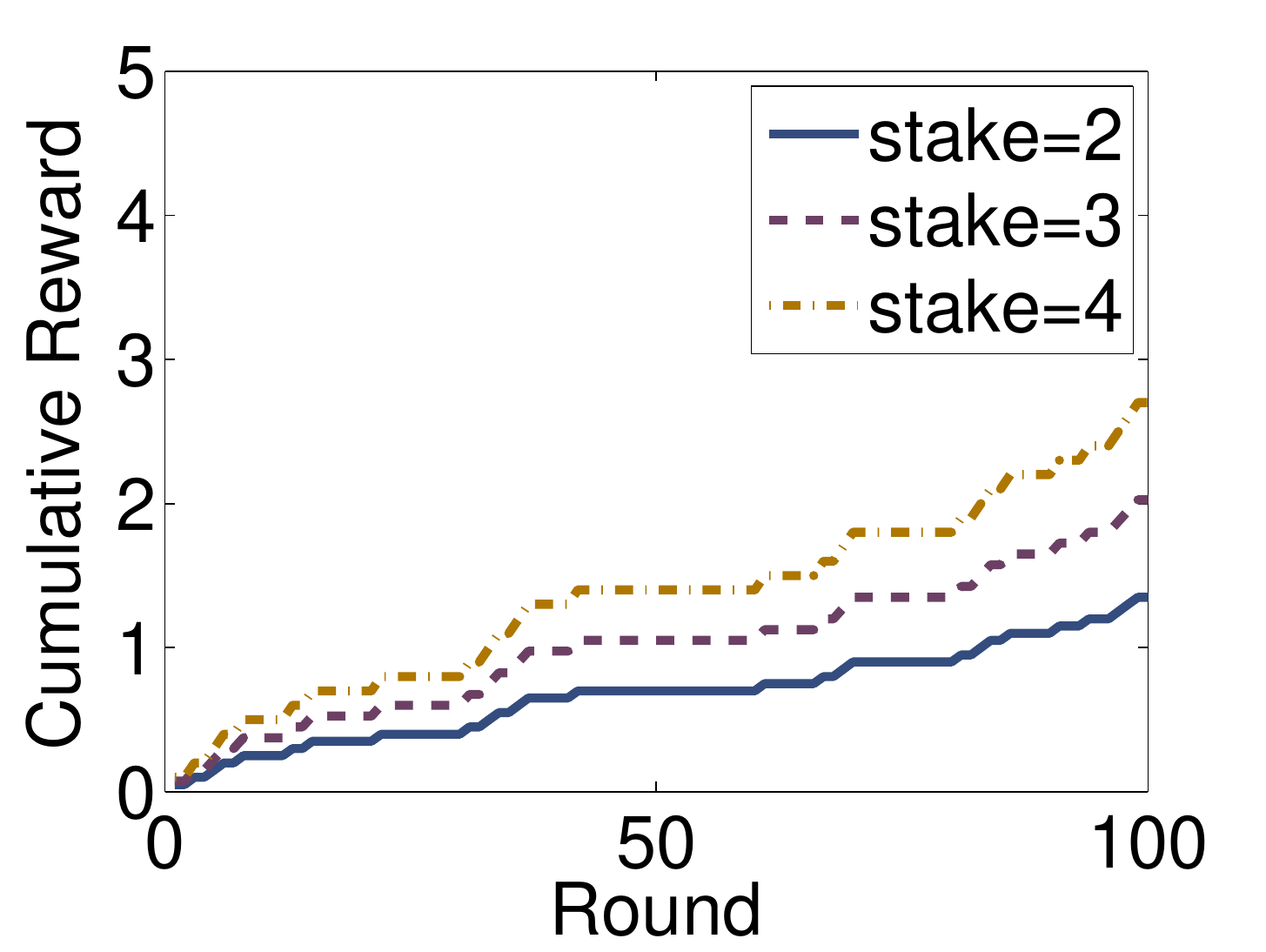}
\captionsetup{font={small}}
\caption{Evolution of the voter's cumulative reward with different numbers of stake.}
\label{lab3}
\end{figure}

\begin{figure}
\centering
\subfigure{
\begin{minipage}[t]{0.25\linewidth}
\centering
\includegraphics[width=1.1\textwidth]{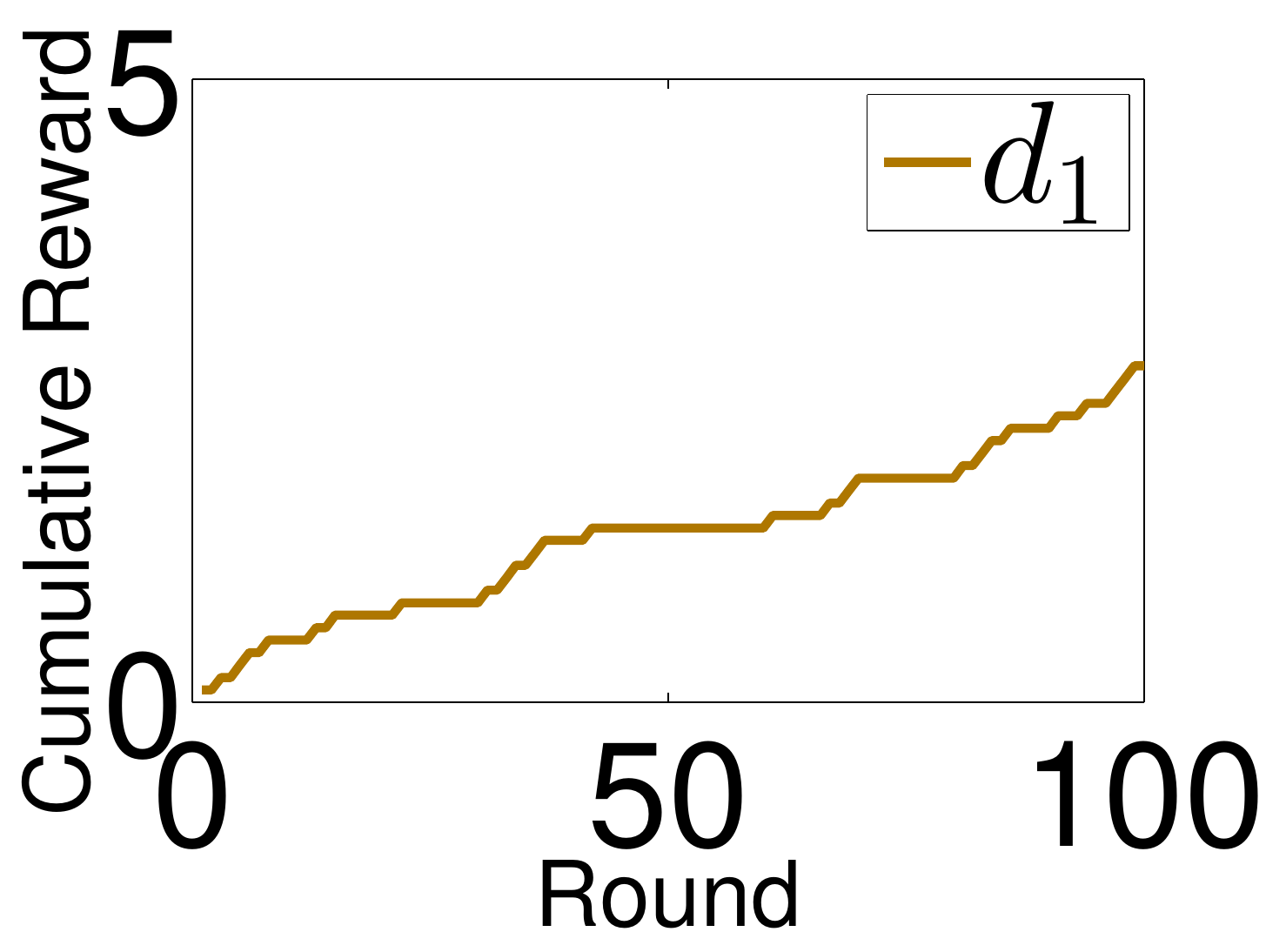}
\end{minipage}
}
\subfigure{
\begin{minipage}[t]{0.25\linewidth}
\centering
\includegraphics[width=1.1\textwidth]{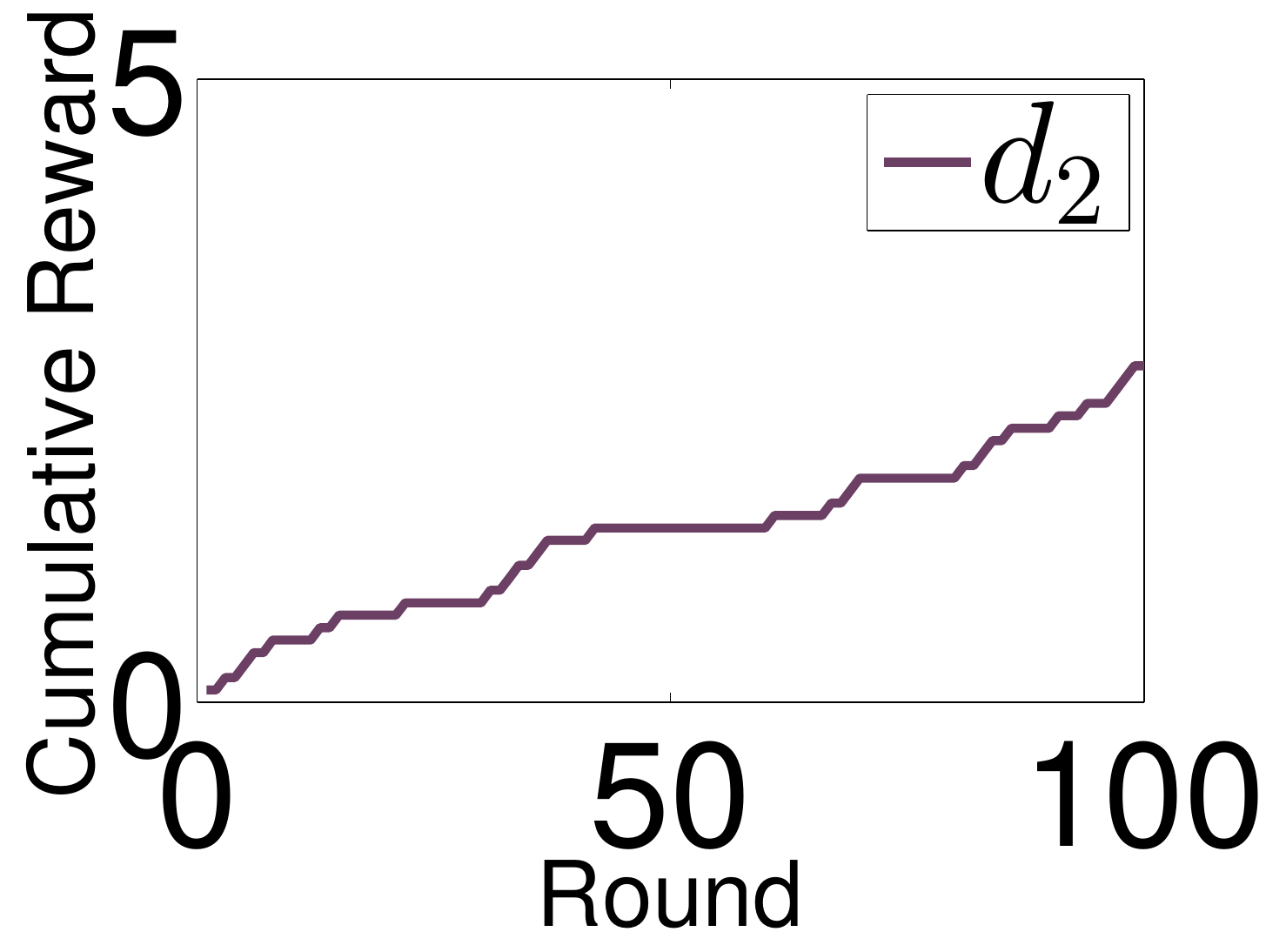}
\end{minipage}
}
\subfigure{
\begin{minipage}[t]{0.25\linewidth}
\centering
\includegraphics[width=1.1\textwidth]{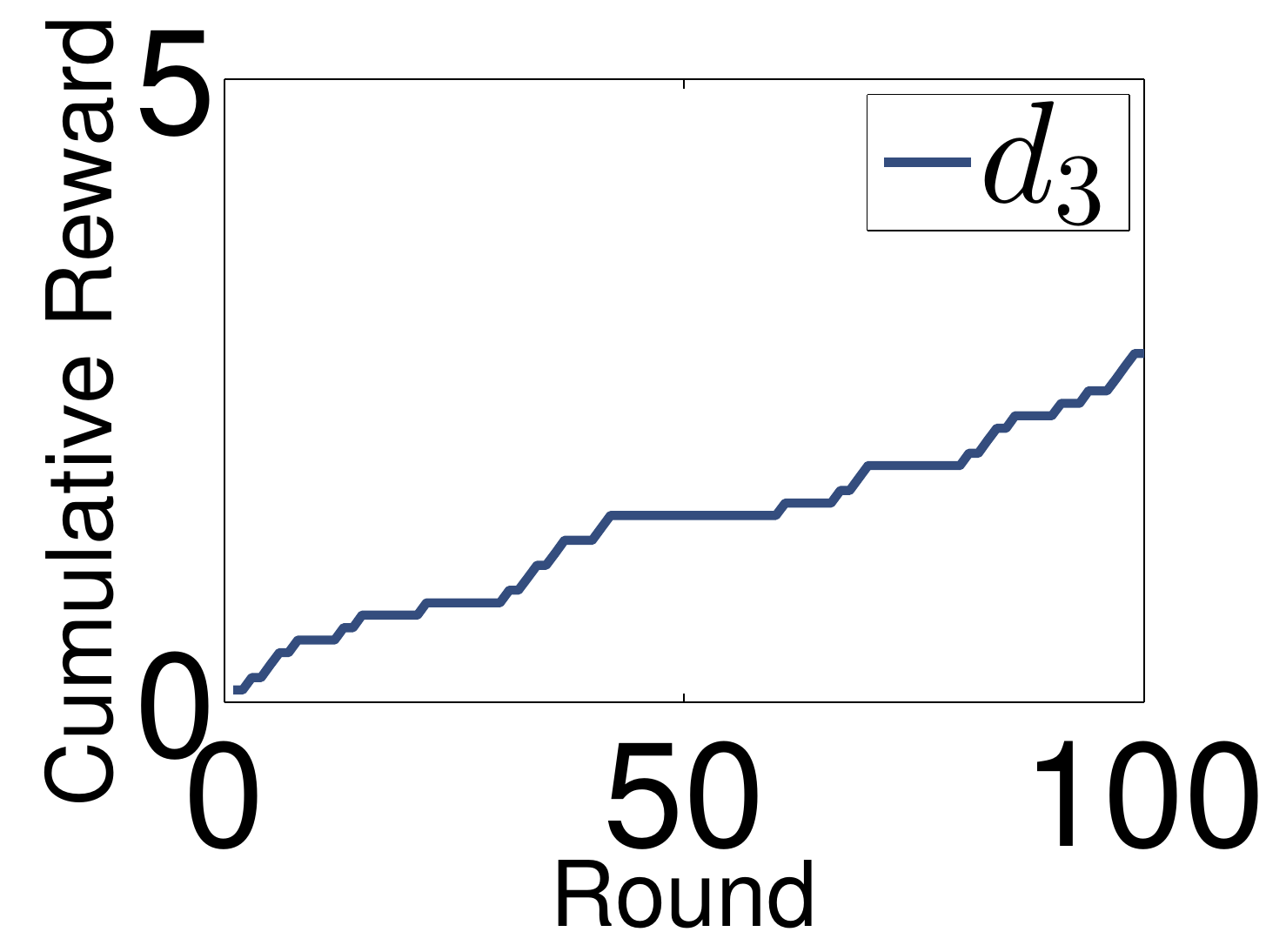}
\end{minipage}
}
\centering
\captionsetup{font={small}}
\caption{Evolution of the voter's cumulative reward with different availability functions.}
\label{lab4}
\end{figure}

In Fig. \ref{lab4}, we plot the evolution of the voter's reward with different unavailability functions $d_1$, $d_2$, and $d_3$. It is clear that all these three functions make the voter's reward present the same evolution trend. That is, the decline rate of availability function does not have a significant impact on the changing rate of the voter's reward. This is in accordance with the requirement that the difference of availability functions should not affect whether the candidate can be selected as super nodes or not, which further determines that the voter's reward will not be influenced.

\subsection{Effectiveness of Our Proposed Voting Consensus Algorithm}
\label{s_6_1}
In this subsection, we investigate the effectiveness of our proposed voting consensus mechanism through comparing the voting results without trustworthiness evaluation and the voting results with trustworthiness evaluation 
in one round of voting. The number of super nodes that need to be elected is set to $K=5$ or $21$, where the corresponding voting results are reported in Figs. \ref{lab1_1} and \ref{lab1_2}, respectively. Note that in the case of $K=5$, as only a few candidates with higher capability to be elected as super nodes are among the first 10 candidates, we report the voting results of them in Fig. \ref{lab1_1} for better presentation even there are 50 candidates in total.

Before conducting the simulation experiment, we rank the capability  of candidate $j$, i.e., $V_j$, in the current round according to their unavailability probability, successful mining rate and community rewards so as to obtain the objective ranking results in the absence of collusion, bribery, inertia and other negative factors, which are reported in the second column of TABLE \ref{tab1} and that of TABLE \ref{tab2}. And the voting results without trustworthiness evaluation in Figs. \ref{lab1_1}(a) and \ref{lab1_2}(a) are numerically presented in the third column of TABLE \ref{tab1} and that of TABLE \ref{tab2}, respectively. Similarly, the voting results with trustworthiness evaluation shown in Figs. \ref{lab1_1}(b) and \ref{lab1_2}(b) are also reported in the fourth column of TABLE \ref{tab1} and that of TABLE \ref{tab2}, respectively.

According to Fig. \ref{lab1_1}(a), in the case of selecting 5 super nodes, the malicious candidate No. 4 bribes some voters to make his total number of votes greater than that of No. 7 in the voting process without trustworthiness evaluation. While after trustworthiness assessment of each voter based on our proposed algorithm, it can be seen that the malicious competition behavior of No. 4 has been successfully suppressed as shown in Fig. \ref{lab1_1}(b), and the ultimately-elected super nodes are exactly what we expect as reported in the fourth column of TABLE \ref{tab1}.
When $K=21$, we assume that candidates No. 1 to No. 4 bribe voters in this round of voting.
In fact, it is unnecessary for candidate No. 3 to bribe any voter as he has the highest real capability to be elected. Through comparing Figs. \ref{lab1_2}(a) and \ref{lab1_2}(b), one can find that the  voting results with  trustworthiness evaluation can definitely meet the basic requirement of selecting super nodes with higher capability; and candidates No. 1, No. 2 and No. 4 are successfully suppressed. The final 21 super nodes are presented in the fourth column in TABLE \ref{tab2}.
All the above experimental results demonstrate that our proposed mechanism can successfully prohibit malicious competition behaviors, such as the wrong election and bribery election, in the voting process. And the underlying reason is that with the help of trustworthiness evaluation for each voter, our proposed voting consensus mechanism can frustrate candidates to bribe and eliminate wrong behaviors of voters.

\begin{figure}[t]
\centering
\subfigure[No trustworthiness evaluation.]{
\begin{minipage}[t]{0.4\linewidth}
\centering
\includegraphics[width=1.1\textwidth]{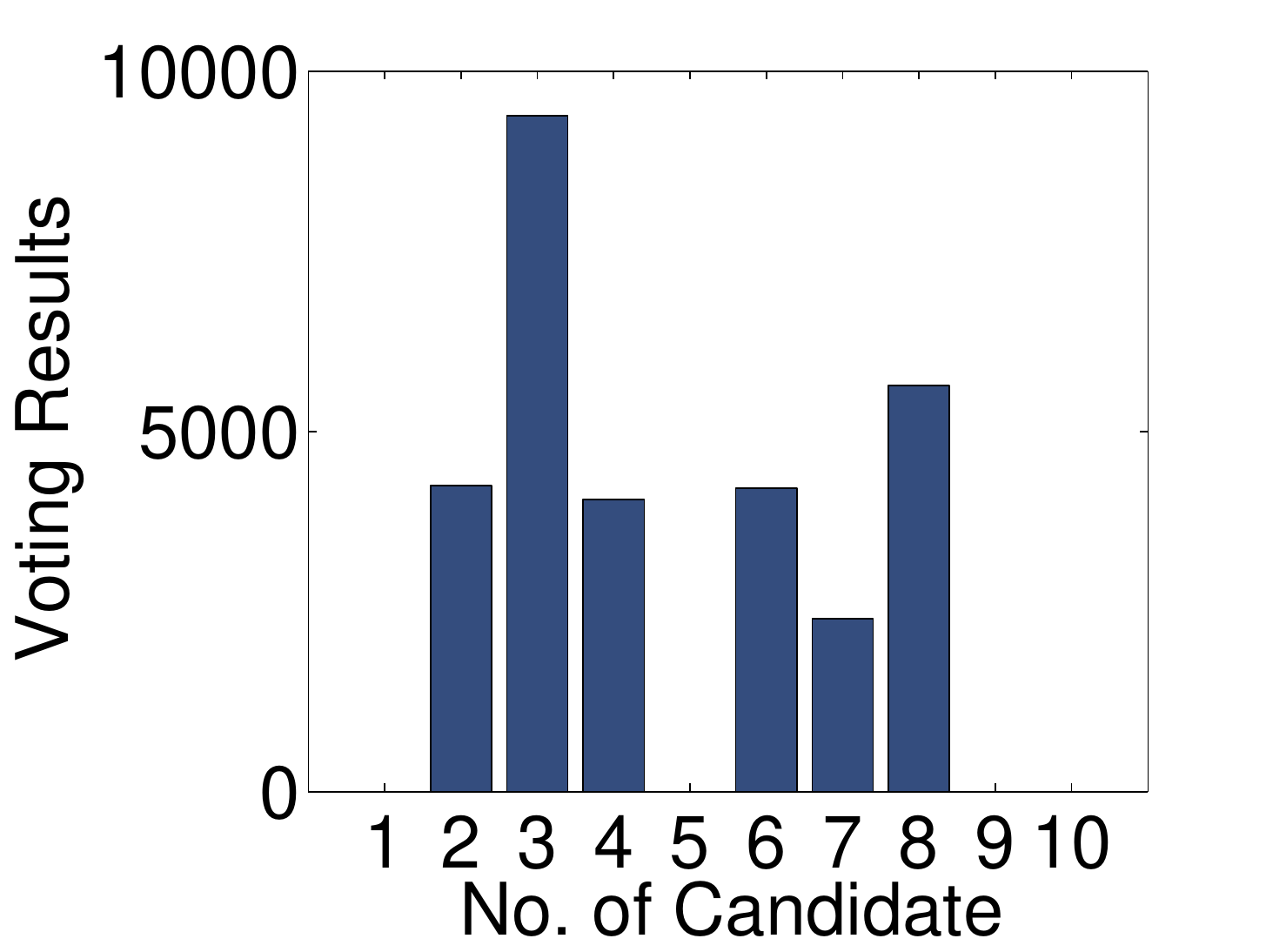}
\end{minipage}
}
\subfigure[With trustworthiness evaluation]{
\begin{minipage}[t]{0.4\linewidth}
\centering
\includegraphics[width=1.1\textwidth]{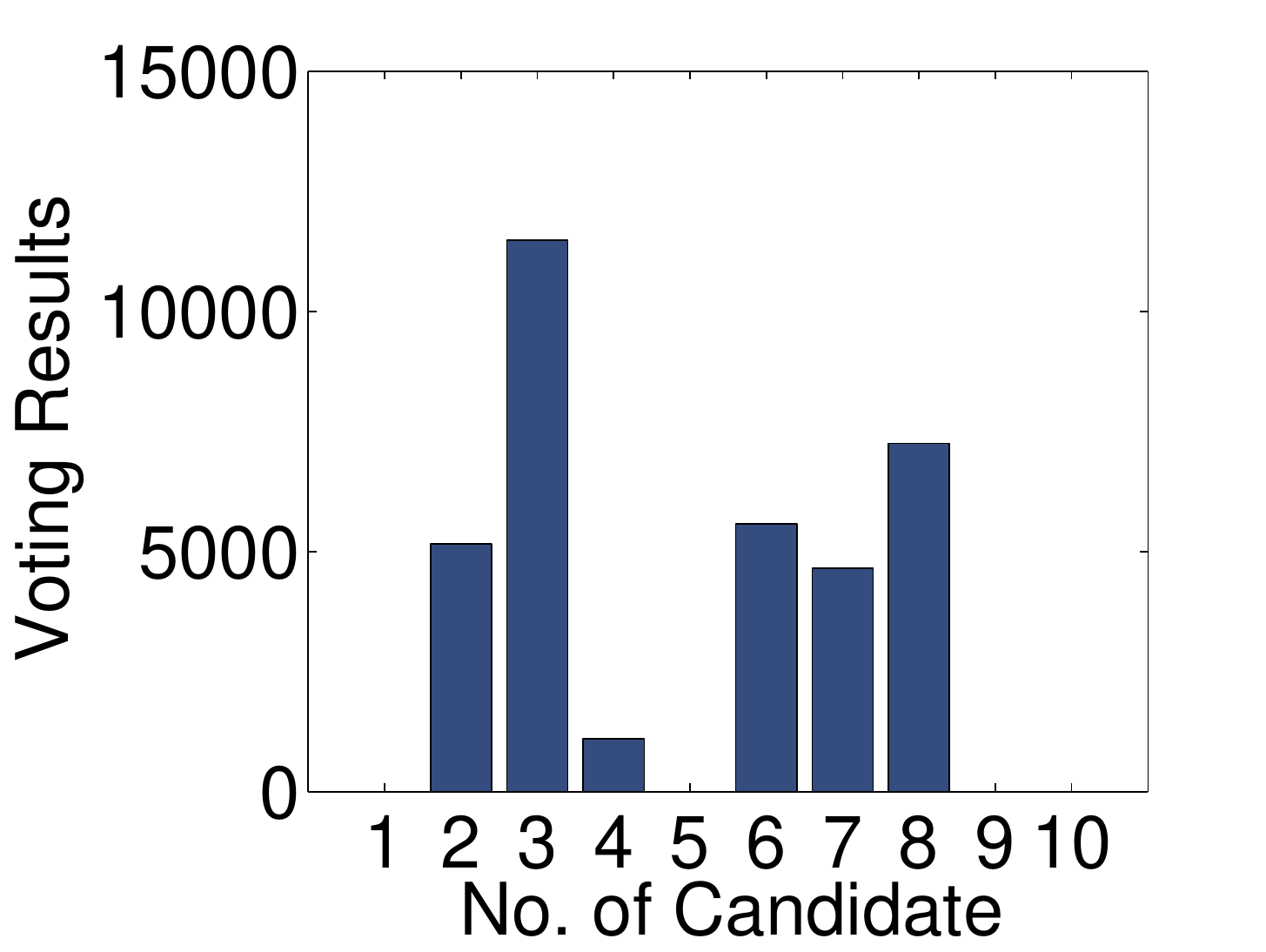}
\end{minipage}
}
\centering
\captionsetup{font={small}}
\caption{Effectiveness of the voting consensus algorithm when $K=5$.}
\label{lab1_1}
\end{figure}

\begin{table}[t]
\setlength{\abovecaptionskip}{0.cm}
\setlength{\belowcaptionskip}{0cm}
\caption{Ranking results when $K=5$.}
\begin{center}
\begin{tabular}{|p{0.8cm}<{\centering}|p{1.1cm}<{\centering}|p{2cm}<{\centering}|p{2cm}<{\centering}|}
\cline{1-4}
\textbf{No. of Super Nodes} & \textbf{Capability Ranking}& \textbf{Voting without Trustworthiness Evaluation}& \textbf{ Voting with Trustworthiness Evaluation} \\
\hline
1 & 3 & 3 & 3 \\
2 & 8 & 8 & 8 \\
3 & 2 & 2 & 6 \\
4 & 6 & 6 & 2 \\
5 & 7 & 4 & 7 \\
\hline
\end{tabular}
\label{tab1}
\end{center}
\end{table}

\begin{figure}[t]
\centering
\subfigure[No trustworthiness evaluation.]{
\begin{minipage}[t]{0.4\linewidth}
\centering
\includegraphics[width=1.1\textwidth]{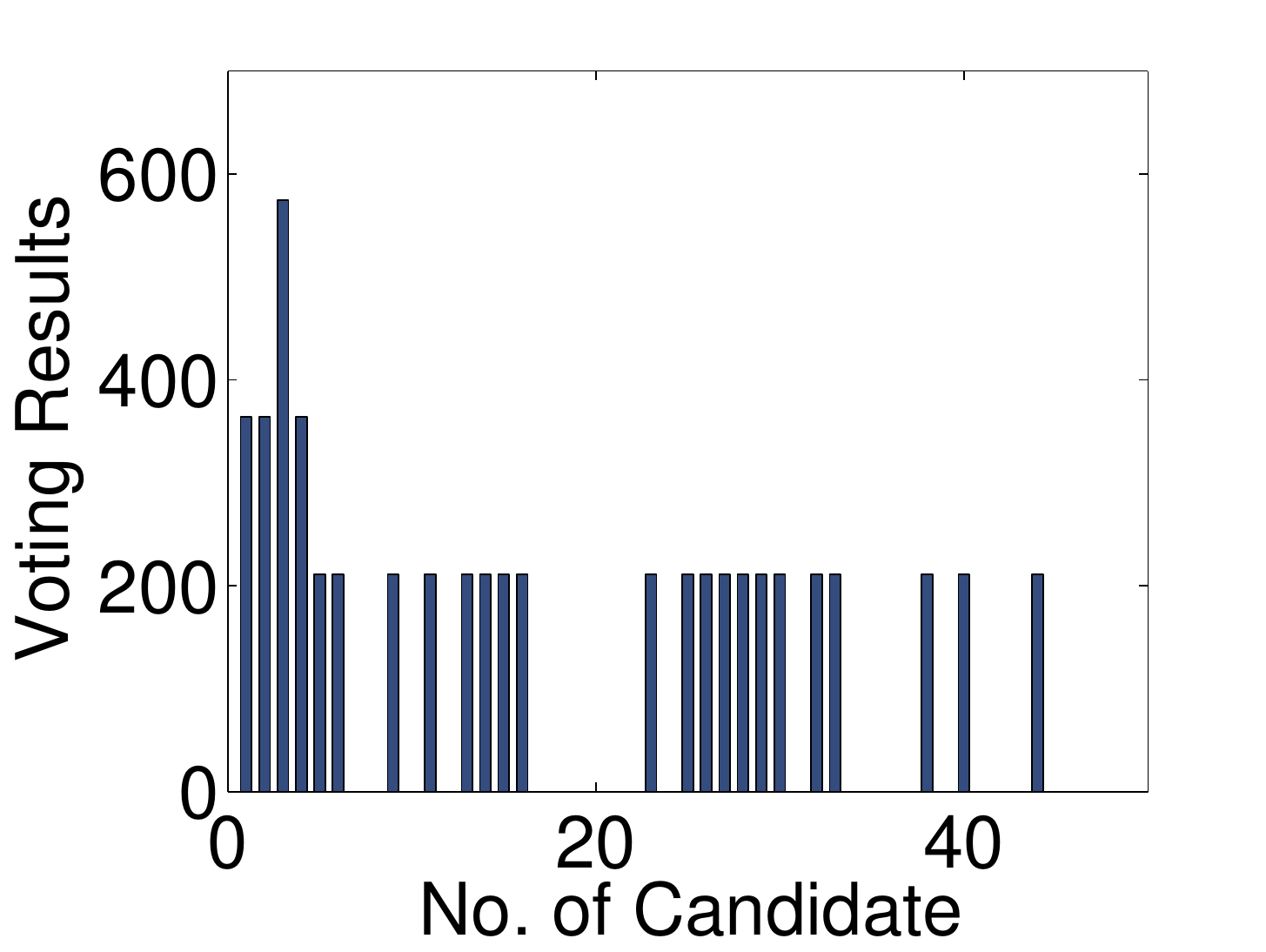}
\end{minipage}
}
\subfigure[With trustworthiness evaluation.]{
\begin{minipage}[t]{0.4\linewidth}
\centering
\includegraphics[width=1.1\textwidth]{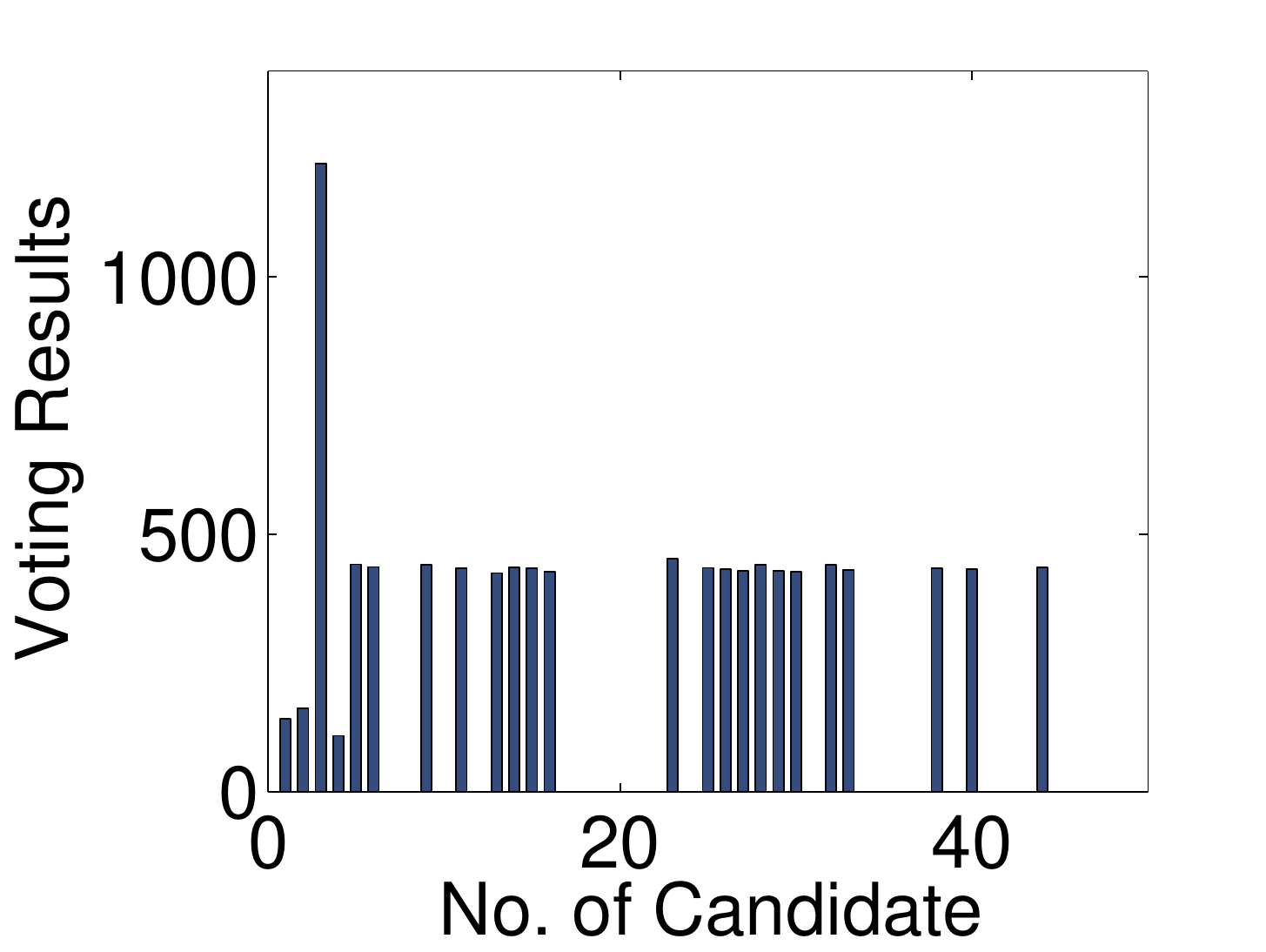}
\end{minipage}
}
\centering
\captionsetup{font={small}}
\caption{Effectiveness of the voting consensus algorithm when $K=21$.}
\label{lab1_2}
\end{figure}

\begin{table}[t]
\setlength{\abovecaptionskip}{0cm}
\setlength{\belowcaptionskip}{0cm}
\caption{Ranking results when $K=21$.}
\begin{center}
\begin{tabular}{|p{0.8cm}<{\centering}|p{1.1cm}<{\centering}|p{2cm}<{\centering}|p{2cm}<{\centering}|}
\cline{1-4}
\textbf{No. of Super Nodes} & \textbf{Capability Ranking}& \textbf{Voting without Trustworthiness Evaluation}& \textbf{ Voting with Trustworthiness Evaluation} \\
\hline
1 & 3 & 3 & 3 \\
2 & 23 & 1 & 23 \\
3 & 5 & 2 & 5 \\
4 & 28 & 4 & 28 \\
5 & 9 & 5 & 9 \\
6 & 32 & 6 & 32 \\
7 & 6 & 9 & 6 \\
8 & 44 & 11 & 44 \\
9 & 14 & 13 & 14 \\
10 & 25 & 14 & 25 \\
11 & 38 & 15 & 38 \\
12 & 15 & 16 & 15 \\
13 & 11 & 23 & 11 \\
14 & 40 & 25 & 40 \\
15 & 26 & 26 & 26 \\
16 & 33 & 27 & 33 \\
17 & 29 & 28 & 29 \\
18 & 27 & 29 & 27 \\
19 & 30 & 30 & 30 \\
20 & 16 & 32 & 16 \\
21 & 13 & 33 & 13 \\
\hline
\end{tabular}
\label{tab2}
\end{center}
\end{table}

\section{Related Work}\label{sec:related}
Generally speaking, the existing consensus algorithms in Blockchain can be divided into two categories: \textit{proof-based} and \textit{voting-based}.

In the proof-based  consensus algorithms,  only when the nodes execute sufficient proof, can they gain the right to mine a new block for rewards. PoW \cite{Nakamoto08} and PoS \cite{King12} are two  classic proof-based consensus algorithms as mentioned in Section \ref{into}. Taking advantage of the traits of  PoW and PoS,   researchers proposed  several variants,
such as the proof of activity  (PoA) \cite{Bentov14}, the proof of capacity (PoC) \cite{Burstcoin14}, and the proof of elapsed time (PoET) \cite{poet17,Chen17}.

In PoA \cite{Bentov14}, a mined block needs to be signed by $N$ miners to become valid. Thus, the difficulty of mining a block depends on the fraction of online stakeholders instead of the number of coins in PoS. Besides, stakes can also be represented by capacity \cite{Burstcoin14}, where miners allocate large capacity in the local hard drive for mining, indicating their interest in mining the block.
In PoET \cite{poet17,Chen17}, the leader is elected through a lottery-based election model, which is required to be run in a trusted execution environment (TEE). In detail, each validator takes advantage of  TEE to generate a waiting time randomly, where the one with the shortest waiting time becomes the leader, ensuring that the leader is fairly selected.

The proof-based consensus algorithms summarized above suffer from a lot of inherent flaws.
PoW wastes too much electrical energy and is not fair for poor miners who cannot afford advanced hardware. And with the appearance of pool mining, more security issues arise or worsen, such as selfish mining \cite{eyal2018majority}, double spending \cite{karame2012double}, block withholding attack \cite{hu2019game}, etc. Because of the latency of new block creation and confirmation,  PoW is unsuitable for real-time payment.
Even though PoS is more efficient and fairer compared to PoW, it still has to calculate a meaningless nonce value and can lead to the situation that the rich get richer.
Meanwhile, the proof-based consensus algorithms are heavily dependent on dedicated hardware for solving the difficult puzzle or using TEE in PoET.

In the voting-based consensus algorithms, nodes know the global information and exchange messages easily so that  the distributed ledger can be appended by not only the miner but almost all nodes through voting. There are two typical classes of voting-based consensus, i.e., the practical Byzantine fault tolerance (PBFT) \cite{Castro99} and DPoS \cite{Larimer14}. 
PBFT \cite{Castro99} is a consensus algorithm used to solve a classic problem called Byzantine Generals (BG) problem, by which $3f+1$ nodes can achieve consensus in the presence of up to $f$ malicious nodes.
It is commonly employed in the permissioned blockchain platform, where the number of nodes is usually limited to 20 as the overhead of messaging increases significantly when the number of replicas increases \cite{Baliga17}.
Based on this, many variants were proposed to improve the consensus performance further.
Raft \cite{Ongaro14} improves PBFT to allow
$f$ malicious nodes in a network of $2f+1$ nodes for achieving consensus.
Tendermint \cite{kwon2014tendermint} optimizes PBFT with the reduced number of voting rounds.
The ripple protocol consensus algorithm (RPCA) \cite{schwartz2014ripple} used in XRP \cite{chase2018analysis} mitigates the requirement of synchronous communication through achieving the federated Byzantine agreement (FBA) in collectively-trusted subnetworks. And Stellar \cite{mazieres2015stellar} refines the implementation of FBA in a provably-safe way.
On the other hand, DPoS \cite{Larimer14,Nguyen18} uses the stake as a proof for voting rather than getting the chance to mine a new block in blockchain, which can be regarded as a combination of voting-based and proof-based consensus.

\section{Conclusion}\label{sec:conclusion}
In this paper, we propose an uncertainty- and collusion-proof voting consensus mechanism, which involves the selection pressure-based voting consensus algorithm to deter wrong elections and adopts an incentive compatible scoring rule  for evaluating the trustworthiness of voting so as to avoid bribery elections.
In addition, we take advantage of the large deviation theory to theoretically analyze the proposed voting consensus mechanism, based on which we  can draw the following conclusions: 1) the growth rate of the voting failure rate can fall off with the increase of  the selection valve more sensitively than with that of the selection pressure; 2)
 the effective selection valve  decreases logarithmically with the increase of the voting failure tolerance degree  and the decrease of the selection pressure; 3) the increase of either the  voting failure tolerance degree  or the selection valve  will lead to the drop of  the effective expectation of  merit, which will remain stable after  declining  to a certain level. Extensive simulations verify the effectiveness of the proposed
 voting consensus mechanism.

\bibliography{reference}
\bibliographystyle{IEEEtran}

\end{document}